\documentclass[lettersize,journal]{IEEEtran}
\usepackage{subcaption} 
\usepackage{booktabs} 
\usepackage{array}    
\usepackage{amsfonts,amssymb,amsbsy,amsmath,paralist,bm,ifthen,color}
\usepackage{graphicx}
\usepackage[noadjust]{cite}
\usepackage[numbers,sort&compress,square]{natbib}
\usepackage{cases,booktabs}
\usepackage{flushend}
\usepackage{algorithmic,algorithm}
\usepackage{lipsum}
\usepackage{amsthm}
\makeatother

\newtheorem{Theorem}{Theorem}

\newcommand\Ac{\ensuremath{{\mathcal{A}}}}

\newcommand\Kc{\ensuremath{{\mathcal{K}}}}
\newcommand\Mc{\ensuremath{{\mathcal{M}}}}
\newcommand\Uc{\ensuremath{{\mathcal{U}}}}

\newcommand\wb{\ensuremath{{\bf w}}}

\newcommand\ub{\ensuremath{{\bf u}}}

\newcommand\Ib{\ensuremath{{\bf I}}}
\newcommand\Ic{\ensuremath{{\mathcal{I}}}}

\newcommand\Vb{\ensuremath{{\bf V}}}

\DeclareMathOperator*{\argmax}{argmax}

\graphicspath{{figures/}}

\definecolor{applegreen}{rgb}{0.55, 0.71, 0.0}

\begin{document}
	
	\title{Distributed User Scheduling in Multi-Cell MIMO O-RAN with QoS Constraints}

    	\author{Tenghao~Cai, Lei~Li,  \textit{Member, IEEE} and Tsung-Hui~Chang,  \textit{Fellow, IEEE} \vspace{-0.6cm}
		
		\thanks{ T. Cai is with the School of Science and Engineering, The Chinese University of Hong Kong, Shenzhen (CUHK-Shenzhen), China, and with the Shenzhen Research Institute of Big Data (SRIBD). L. Li and T.-H. Chang are with School of Artificial Intelligence, CUHK-Shenzhen and with the SRIBD.  (email: 221019048@link.cuhk.edu.cn, lei.ap@outlook.com, tsunghui.chang@ieee.org).}

	}

	\maketitle

	\begin{abstract}
		Distributed scheduling is essential for open radio access network (O-RAN) employing advanced physical-layer techniques such as multi-user MIMO (MU-MIMO), carrier aggregation (CA), and joint transmission (JT). This work investigates the multi-component-carrier (multi-CC) resource block group (RBG) scheduling in MU-MIMO O-RAN with both JT and non-JT users. We formulate a scheduling optimization problem to maximize throughput subject to user-specific quality of service (QoS) requirements while ensuring consistent allocations across cooperating O-RAN radio units (O-RUs) required by JT transmission. The strong variable coupling, non-convexity, and combinatorial complexity make the problem highly challenging. To tackle this, we extend the eigen-based zero-forcing transceiver design to JT users and leverage massive MIMO asymptotic properties to derive a tractable, separable rate approximation. Building on this, we develop two solutions: a centralized block coordinate descent benchmark and a distributed scheduler aligned with the O-RAN architecture. The proposed distributed scheme achieves near-centralized performance with only one round of lightweight coordination among cells, significantly reducing complexity and delay. Extensive simulations validate that our distributed scheduler achieves high scalability, fast convergence, and better QoS satisfaction rate in large-scale MU-MIMO networks.
	\end{abstract}
	\vspace{-0.3cm}
\section{Introduction}
	\label{sec:intro}
	\vspace{-0.05cm}
	The growing diversity of user services and rising demands for quality of service (QoS) have placed unprecedented pressure on wireless networks. To satisfy these requirements, multi-user multi-input multiple-output (MU-MIMO) systems have become a cornerstone of 5G and beyond, enabling base stations (BSs) to serve multiple user equipments (UEs) concurrently within the same time-frequency resources through spatial multiplexing. In parallel, to further boost capacity, carrier aggregation (CA) \cite{lin2022carrier,zhou2025multi} allows UEs to simultaneously access multiple component carriers (CCs) operating at distinct center frequencies. Each CC is divided into resource block groups (RBGs), providing finer spectrum granularity. 
	However, even with enhanced frequency-domain flexibility and spatial multiplexing, cell-edge UEs still suffer from weak signal strength and strong inter-cell interference (ICI). To address this issue, joint transmission (JT) was introduced to allow multiple BSs to coherently serve the same UE to improve the received signal quality \cite{marsch2011coordinated}. 
	To fully harness these benefits of MU-MIMO, multi-CC, and JT, it is essential to design an efficient scheduler that allocates the RBGs across multiple CCs to UEs in multiple cells to maximize system throughput while meeting QoS requirements. However, designing such a scheduler is notoriously non-trivial due to the following complexities \cite{he2022cross}:
	\vspace{-0.05cm}
	\begin{itemize}
		\item \textbf{Scheduling Consistency}:   For UEs jointly served by multiple BSs, referred to as JT-UEs, the scheduler must ensure consistent RBG allocations across all coordinating BSs for each CC. 
		
		\item \textbf{QoS Guarantee}: For UEs with stringent QoS requirements, the scheduler design needs to coordinate across BSs and CCs to jointly satisfy the UE's QoS requirement. This is particularly challenging since communications across BSs and CCs are strictly limited. 
		
		\item \textbf{Coupled Scheduling and Transmit Beamforming (TBF) Design}: Mitigating ICI requires multi-user TBF, whose design depends on the scheduling pattern. Yet, the scheduling itself depends on the interference induced by TBF, creating a tightly coupled optimization problem.
	\end{itemize}
	Accordingly, MU-MIMO scheduling exhibits strong non-convexity and coupling across CCs, BSs, and JT-UEs, yielding an NP-hard integer optimization problem \cite{denis2021improving}.
	
	Moreover, an efficient scheduler must be computation-fast and scalable with respect to the network size. \textit{Distributed scheduling} addresses this by offloading RBG allocation across network nodes and leveraging multi-core parallelism within each node. To this end, the scheduling design must align with the network architecture well. Among various wireless network architectures, the open radio access network (O-RAN), with its promoted openness and modularity, has emerged as a paradigm-shifting framework for next-generation wireless systems \cite{polese2023understanding}. In O-RAN’s disaggregated architecture, the open radio unit (O-RU) is responsible for RF signal transmission and reception, while the open distributed unit (O-DU), connected to multiple O-RUs, performs baseband processing and scheduling functions.
	In practical deployments, an O-DU is commonly realized as a computing platform equipped with multiple processing units (PUs) or accelerator cards, interconnected through high-speed on-board links to support parallel processing and lightweight inter-unit coordination \cite{10969847}.  	
    Under this architecture, beyond the intrinsic non-convexity and tight coupling of MU-MIMO scheduling across cells, carriers, and JT-UEs, the scheduler must effectively exploit the internal parallelism of the O-DU while maintaining scheduling consistency. Meanwhile, the information exchange among computational units is latency-sensitive and cannot be arbitrarily frequent. This makes the corresponding algorithm design more challenging.
    

	Although scheduling optimization in multi-cell networks has been widely studied, distributed multi-cell, multi-CC scheduling for MU-MIMO O-RAN that explicitly accounts for JT and QoS constraints remains unexplored. For example, algorithms in \cite{sun2010genetic,khan2020optimizing,yu2012coordinated,he2022cross,denis2021improving} considered joint optimization of user scheduling and resource allocation, but they focused on centralized processing.  Moreover, the works \cite{khan2020optimizing,yu2012coordinated,he2022cross,denis2021improving} were restricted to the MISO scenario. In contrast, the distributed schemes in  \cite{li2014multicell,li2015multicell,gamvrelis2022slinr,chen2023om} ignored QoS constraints, rendering them unsuitable for burst-demand services such as extended reality. While \cite{antonioli2020decentralized} considered QoS requirements, it took multiple rounds of information exchange between BSs due to the iterative TBF design, incurring substantial signaling overhead and latency that hinder practical deployment.
    
    In this work, we study cross-CC scheduling in MU-MIMO O-RAN involving both JT-UEs and NJT-UEs (i.e., UEs served by a single BS), and propose a highly parallel scheduling framework. By judiciously decomposing the scheduling problem into multiple subproblems and exploiting the parallel processing capability of the O-DU, our proposed scheme achieves distributed scheduling with only a single round of information coordination among multiple PUs of the O-DU, significantly reducing the cross-PU signaling overhead while effectively balancing spectral efficiency and diverse QoS requirements. The main contributions are summarized as follows:
	\vspace{-0.05cm}
    \begin{enumerate}
	\item \textbf{Multi-Cell Cross-CC Scheduling Formulation for MU-MIMO O-RAN}:  
	We consider general multi-cell cross-CC scheduling in MU-MIMO O-RAN, including both JT-UEs and NJT-UEs with heterogeneous QoS requirements. Given practical computational constraints where iterative TBF solutions are often prohibitive, we adopt eigen-based zero-forcing TBF (EZF-TBF) with a closed-form transceiver design, and formulate a scheduling optimization problem that jointly maximizes throughput and satisfies QoS constraints. Unlike prior work restricted to single-cell scheduling \cite{9488684}, single-CC operation \cite{antonioli2020decentralized}, or scheduling without explicit QoS guarantees \cite{chen2023om}, our model captures the intrinsic complexity of practical scheduling and incorporates inter-cell JT scheduling consistency and cross-CC QoS coupling.

	\item \textbf{Efficient Problem Approximation}:  
	To handle the highly non-convex problem with strong coupling between binary scheduling variables and intermediate TBF variables, we develop a novel approximation scheme. First, by exploiting the structural properties of EZF-TBF and the characteristics of ICI, we derive closed-form approximate rate expressions for JT-UEs, substantially alleviating cross-cell coupling. Leveraging massive MIMO asymptotics, we further remove the dependence on intermediate TBF variables, yielding simplified rate expressions that are amenable to distributed and parallel optimization.

	\item \textbf{Distributed Scheduling with Lightweight Coordination}:
	Building on the simplified problem, we develop a distributed scheduling scheme that fully exploits the parallel computing capabilities of multiple PUs within the O-DU while requiring only a single round of coordination. The scheme consists of three stages: decomposition, coordination, and refinement. First, the network-wide scheduling problem is decomposed into parallel sub-tasks executed across computational cores within the PUs, enabling scalable decision-making. Next, limited coordination occurs at two levels: within each PU, cores exchange local information to align NJT-UE QoS constraints across CCs; in parallel, a dedicated PU performs inter-PU coordination to ensure consistent JT-UE scheduling across cells. Finally, each core refines its local NJT-UE scheduling based on the feedback, further improving performance through localized optimization.

	\item \textbf{Extensive Validation}:  
	Extensive simulations demonstrate that the proposed distributed scheduler achieves near-centralized performance in terms of throughput and QoS satisfaction, while reducing computational time by more than an order of magnitude, making it well-suited for large-scale O-RAN deployments.
\end{enumerate}

	\vspace{-0.4cm}
	\subsection{Related Work}
	\vspace{-0.1cm}
	\label{subsec:related}
	As a cornerstone of modern wireless networks, the multi-cell user scheduling problem is inherently complex due to intertwined factors such as resource coordination and heterogeneous QoS requirements. The scheduling optimization is typically cast as an integer programming problem with combinatorial complexity and is well known to be NP-hard. Accordingly, exhaustive search with a prohibitive computational burden is impractical for real-world systems. {Similarly, scheduling schemes based on commercial solvers like Gurobi often suffer from prohibitive computational complexity and poor scalability \cite{bischoff2024real}, particularly under explicit QoS constraints.}
	Consequently, researchers have explored various centralized/distributed strategies, including relaxation-based methods \cite{denis2021improving}, genetic algorithms \cite{sun2010genetic}, branch-and-bound \cite{yu2012coordinated}, and Hungarian algorithm \cite{khan2020optimizing}. 
	In parallel, lightweight heuristic schemes such as proportional fairness \cite{jorswieck2009throughput} have been adopted, offering low-complexity implementations at the expense of performance loss.
	
	\textbf{Centralized scheduling} typically relies on a central unit that collects global information and optimizes network-wide decisions. In \cite{sun2010genetic}, genetic algorithms were employed to evolve user scheduling patterns to maximize the throughput in multi-cell MIMO systems under successive interference cancellation decoding strategies. The work \cite{khan2020optimizing} integrated fractional programming with the Hungarian algorithm for multi-band scheduling, alternately optimizing user assignments and TBF to maximize the weighted sum rate. 
    \textit{While effective, these designs share a critical limitation:} they predominantly targeted throughput maximization while neglecting explicit per-user QoS constraints.
	Compared with unconstrained formulations, incorporating QoS constraints makes multi-cell scheduling considerably harder. First, the feasible region becomes drastically compressed and typically non-convex. 
     Second, existing methods based on dual-decomposition for these problems often introduce additional optimization variables or parameters, which further increase computational complexity. 
	Third, unlike unconstrained cases where each RB can be optimized independently, 
	the scheduler needs to handle the variable coupling due to QoS constraints and jointly determine allocations across all RBs. This further increases the computational burden. 
	
	In existing literature, \cite{yu2012coordinated} studied multi-subchannel scheduling to maximize the number of scheduled users under QoS constraints, and proposed a joint scheduling and TBF algorithm via branch-and-bound. To enhance the QoS of users under weak coverage, cooperative transmissions have also been explored. The work \cite{he2022cross} proposed alternating optimization of user scheduling and TBF to maximize sum rate under QoS and power constraints in multi-cell JT networks. 
	Further, in the cell-free massive MIMO, \cite{denis2021improving} proposed a grouping-based carrier and power allocation approach under multi-band QoS constraints,  utilizing Lagrangian relaxation for carrier assignment to user clusters and sequential convex approximation for power control.
	However, these centralized approaches inevitably impose a substantial computational load on the central unit, especially in large-scale networks. This makes distributed scheduling schemes with lower complexity and better scalability particularly important.

	\textbf{Distributed approaches} 
    delegate scheduling optimization to individual BSs, thereby alleviating the computation load at a single node. However, the corresponding design faces several fundamental challenges. First, the strong inter-dependencies among variables make the problem difficult to decompose. Second, the limited local information at each node impedes accurate ICI estimation and results in imprecise rate predictions. Third, multi-round coordination among distributed nodes is often required to reach consensus, entailing substantial signaling overhead and delay.
    In the literature, early designs often relied on simple operational assumptions. 
    For instance, \cite{li2014multicell} devised a scheduling approach in a two-cell system, where each cell first selects candidate UEs via local semi-orthogonal user selection (SUS), exchanges these sets with the other BS, and then performs coordinated zero-forcing (ZF) TBF to cancel ICI. This was extended to multiple cells in \cite{li2015multicell} through sequential scheduling, where cells were activated in a predefined order, and each subsequent BS made scheduling decisions based on ICI from previously scheduled cells, albeit with latency scaling linearly with network size. 
    Besides, a leakage-based distributed scheme was proposed in \cite{gamvrelis2022slinr}, where the ICI was approximated by the leakage to the active users in other cells, estimated by a traffic-aware statistical model. While reducing signaling, this scheme cannot reliably handle explicit QoS constraints due to the inherent inaccuracy of leakage-based approximations for rate prediction.
    In addition, \cite{antonioli2020decentralized} introduced rate relaxation variables (RRVs) to allow soft violations of rate requirements and developed a primal–dual algorithm enabling distributed implementation via iterative exchange of dual variables and interference messages. However, this approach suffered from high signaling overhead due to the iterative message passing, and its RRV-based soft constraints inherently favored users with strong channels, further degrading cell-edge UE rates. 
	In \cite{chen2023om}, a real-time multi-cell MIMO scheduler for O-RAN was developed, leveraging GPU-based parallel processing to jointly optimize RB allocation, MCS selection, and TBF. 
	However, this design does not consider QoS constraints, and extending it to handle QoS requirements across multiple RBs remains challenging.
	
	\vspace{-0.2cm}
	\section{System Model and Problem Formulation}
	\vspace{-0cm}
	\subsection{System Model} 
    \vspace{-0cm}
	As shown in Fig. \ref{fig:sys}, we consider an O-RAN  with $M$ cells and $K$ UEs, each equipped with $N_r$ receive antennas. Each cell has an O-RU with $N_t$ antennas at its center, connected to a shared O-DU via fronthaul. The O-DU, responsible for baseband processing and scheduling tasks for the network, is composed of multiple PUs. Each PU contains multiple interconnected computational cores to enable parallel processing, and the PUs themselves are connected via high-speed on-board links to support coordination and information exchange. In practice, a PU can be realized on various hardware platforms, such as multi-core CPUs, system-on-chip (SoC) modules, or FPGA-based accelerators { \cite{10969847}}.
    \begin{figure}[t] 
		\centering	
		{\includegraphics[width=0.48\textwidth]{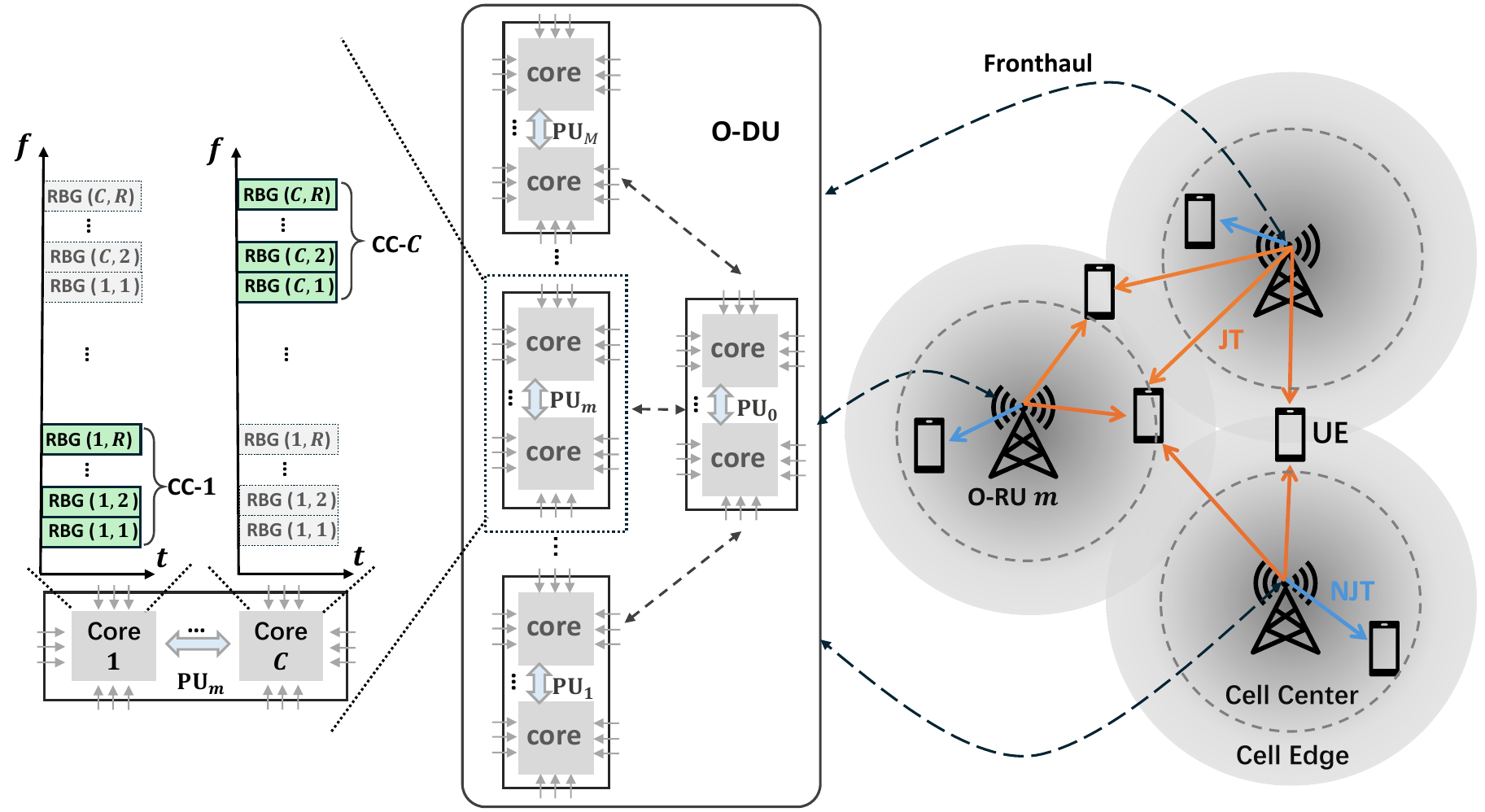}}\vspace{-0.15cm}	
		\caption{System Model} 
		\label{fig:sys} 
	\end{figure} 
    
    The set $\Ac_m$ of UEs associated with O-RU $m$ is assumed known and is partitioned based on their experienced interference into $\Ac_m \triangleq \{\Kc_m, \Uc_m\}$. Here, $\Kc_m$ represents cell-center UEs primarily affected by intra-cell interference, and $\Uc_m$ stands for the set of cell-edge UEs that experience strong ICI. To improve the service quality of cell-edge UEs, multiple O-RUs in adjacent cells employ JT to convert detrimental ICI into coherent useful signals, and these UEs are named JT-UEs. Meanwhile, cell-center UEs are served exclusively by their associated O-RU without JT and are referred to as NJT-UEs.
	
	The network operates under a time-slotted scheduling mechanism. We focus on scheduling within a single interval, during which the total available bandwidth comprises $C$ CCs, each of which occupies $R$ RBGs. Without loss of generality (w.l.o.g.), we assume that the O-DU comprises $M+1$ PUs, each containing $C$ computational cores. For $\text{PU}_1$ to $\text{PU}_M$, core $c$ on $\text{PU}_m$ is responsible for scheduling UEs of cell $m$ over the $c$-th CC, while $\text{PU}_0$ handles inter-cell coordination. Table \ref{tab:notations} lists the key notations used in the paper.
	
	    \begin{table}[t]
		\centering
		\caption{Notations}
		\label{tab:notations}
		\renewcommand{\arraystretch}{1.2}
		\begin{tabular}{@{}p{2.5cm}p{6cm}@{}}  
			\toprule
			\textbf{Notation} & \textbf{Definition} \\
			\midrule
			$\mathcal{M} \triangleq \{1,\ldots,M\}$ & Set of O-RUs  \\
			$\mathcal{K} \triangleq \{1,\ldots,K\}$ & Set of UEs \\
			$\mathcal{I} \triangleq \{1,\ldots,I\}$  & Set of JT-UEs \\
			$\mathcal{B}_i \subseteq \mathcal{M}$ & Subset of serving O-RUs for JT-UE $i$ \\
			$\mathcal{K}_m \subseteq \mathcal{K}$ & Set of cell-center/NJT-UEs in cell $m$ \\
			$\mathcal{U}_m \subseteq \mathcal{I}$ &Set of cell-edge/JT-UEs in cell $m$ \\
			$\mathcal{K}_m^{Q}\subseteq \mathcal{K}_m$ & Set of NJT-UEs with QoS requirement in cell $m$ \\ $\mathcal{\hat{K}}_m^{Q}\subseteq \mathcal{K}_m $ & Set of NJT-UEs w/o QoS requirement in cell $m$\\ 
			$\mathcal{U}_m^{Q}\subseteq \mathcal{U}_m$ & Set of JT-UEs with QoS requirement in cell $m$ \\ $\mathcal{\hat{U}}_m^{Q}\subseteq \mathcal{U}_m $ & Set of JT-UEs w/o QoS requirement in cell $m$\\ 
			$\mathcal{I}^{Q} \subseteq \mathcal{I}$ & Set of JT-UEs with QoS requirement \\
			$\mathcal{\hat{I}}^Q \subseteq \mathcal{I} $ & Set of JT-UEs w/o QoS requirement  \\
			$\mathcal{C} \triangleq \{1,\ldots,C\}$ & Set of CCs \\
			$\mathcal{R} \triangleq \{1,\ldots,R\}$ &Set of RBGs \\
			RBG $(c,r)$ & $r$-th RBG in  the $c$-th CC\\
			$\mathbf{u}^{c,r}_{k} \in \mathbb{C}^{N_r \times 1}$ & Receive combiner for UE $k$ on RBG $(c,r)$\\
			$\mathbf{w}_{m,k}^{c,r} \in \mathbb{C}^{N_t \times 1}$ & TBF from O-RU $m$ to UE $k$ on RBG $(c,r)$\\
			\bottomrule
		\end{tabular}  
	\end{table}

	Denote $b_{m,k}^{c,r} \in \{0,1\}$ as the scheduling variable indicating whether UE $k \in \mathcal{A}_m$ is allocated to RBG $r$ of the $c$-th CC, 
	
	\vspace{-0.35cm}
    \begin{align}
    			b_{m,k}^{c,r}= \begin{cases}1, & \text{if UE $k$ is scheduled      on RBG $(c,r)$,} \\ 0, & \text{otherwise.}\end{cases}
    		\end{align}
    \vspace{-0.35cm}
	
	\noindent Assume that the channel from O-RU $m$ to UE $k$ on RBG $(c,r)$ as $\mathbf{H}_{m,k}^{c,r} \in \mathbb{C}^{N_r \times N_t}$ is available. 
	The transmit data symbol for UE $k$ satisfies $\mathbb{E}[|x_k|^2] = 1$ and is assumed to be mutually independent across different UEs. After applying the TBF $\wb^{c,r}_{m,k}$ and the receive combiner $\ub_k^{c,r}$, the received signal of NJT-UE $k$ in cell $m$ on RBG $(c,r)$ is written as  

    \vspace{-0.3cm}
		\begin{align}\label{received signal NJT_UE}
			{y}_{m,k}^{c,r} 
			 = &(\mathbf{u}^{c,r}_{k})^\mathrm{H}   
			\mathbf{H}_{m,k}^{c,r} \mathbf{w}_{m,k}^{c,r} x_{k} \notag \\
			+ & \sum\nolimits_{n\in \mathcal{M} } \sum\nolimits_{\substack{t \in \mathcal{K}_{n},t \neq k}} (\mathbf{u}^{c,r}_{k})^\mathrm{H}b_{n,t}^{c,r}\mathbf{H}_{n,k}^{c,r} \mathbf{w}_{n,t}^{c,r} x_{t}   \\
			+ & \sum\nolimits_{i \in \mathcal{I}} \sum\nolimits_{\ell \in \mathcal{B}_i}  (\mathbf{u}^{c,r}_{k})^\mathrm{H}b_{\ell,i}^{c,r}\mathbf{H}_{\ell,k}^{c,r} \mathbf{w}_{\ell,i}^{c,r} x_{i} + (\mathbf{u}^{c,r}_{k})^\mathrm{H}\mathbf{n}_k, \notag
		\end{align}
	\vspace{-0.3cm}
	
	\noindent where $\mathbf{n}_k \sim \mathcal{CN}(0, \sigma^2 \Ib_{N_r})$ is the additive white Gaussian noise (AWGN).  The first three terms in \eqref{received signal NJT_UE} represent the desired signal, the interference from co-scheduled NJT-UEs, and the interference from co-scheduled JT-UEs, respectively.

	The received signal ${y}_{i}^{c,r}$ of JT-UE $i$ over RBG $(c,r)$ is
	
	\vspace{-0.4cm}
	
		\begin{align}\label{receive-JT}
			{y}_{i}^{c,r}\! & \!= \!    \sum\nolimits_{\ell \in \mathcal{B}_i} (\mathbf{u}^{c,r}_{i})^\mathrm{H} \mathbf{H}_{\ell,i}^{c,r} \mathbf{w}_{\ell,i}^{c,r}  x_{i}   \notag \\ 
			&   +\! \sum\nolimits_{\ell \in \mathcal{M}}\sum\nolimits_{\substack{t \in \mathcal{K}_\ell}} (\mathbf{u}^{c,r}_{i})^\mathrm{H}b_{\ell,t}^{c,r} \mathbf{H}_{\ell,i}^{c,r} \mathbf{w}_{\ell,t}^{c,r}  x_{t}    \\
			&  +\!\sum\nolimits_{j \in \mathcal{I}, j \ne i} \sum\nolimits_{\ell \in \mathcal{B}_{j}} \! (\mathbf{u}^{c,r}_{i})^\mathrm{H}b_{\ell,j}^{c,r}\mathbf{H}_{\ell,i}^{c,r} \mathbf{w}_{\ell,j}^{c,r} x_{j} \!+ \!(\mathbf{u}^{c,r}_{i})^\mathrm{H}\mathbf{n}_i  \notag,
		\end{align}

	\noindent where the first three terms comprise the desired signal from its serving cells, along with interference from all co-scheduled NJT-UEs and other co-scheduled JT-UEs. 
	
	Then, from \eqref{received signal NJT_UE}, the signal-to-interference-plus-noise-ratio (SINR) of NJT-UE $k \in \Kc_m$ on RBG $(c,r)$ is
	
	\vspace{-0.35cm}
	
		\begin{align}\label{gamma}
			\gamma_{m,k}^{c,r}(\{ b_{m,k}^{c,r},\mathbf{w}_{m,k}^{c,r}\} ) &
			=  |   (\mathbf{u}^{c,r}_{k})^\mathrm{H} \mathbf{H}_{m,k}^{c,r} \mathbf{w}_{m,k}^{c,r} |^2 \notag \\
			&   \!\!\!\!\!\!\!\!\!\! \!\!\!\!\!\!\!\!\!\! \!\!\!\!\!\!\!\!\!\! \!\!\!\!\!\!\!\!\!\! \times\bigr( \sum\nolimits_{n\in \mathcal{M}  } \sum\nolimits_{\substack{t \in \mathcal{K}_{n},t\neq k}} |(\mathbf{u}^{c,r}_{k})^\mathrm{H}b_{n,t}^{c,r}\mathbf{H}_{n,k}^{c,r} \mathbf{w}_{n,t}^{c,r}  |^2\\
			&  \!\!\!\!\!\!\!\!\!\! \!\!\!\!\!\!\!\!\!\!  \!\!\!\!\!\!\!\!\!\!  \!\!\!\!\!\!\!\!\!\! + \sum\nolimits_{i \in \mathcal{I}} \bigr|\sum\nolimits_{\ell \in \mathcal{B}_i} (\mathbf{u}^{c,r}_{k})^\mathrm{H}b_{\ell,k}^{c,r}\mathbf{H}_{\ell,k}^{c,r} \mathbf{w}_{\ell,i}^{c,r}  \bigr|^2+\sigma^2 \bigr)^{-1}, \notag
		\end{align}

    \vspace{-0.1cm}
	\noindent and, from \eqref{receive-JT}, the SINR for JT-UE  $i$ is
	
	\vspace{-0.4cm}
	
		\begin{align}\label{gamma_JT}
			\gamma_{i}^{c,r}(\{ b_{m,k}^{c,r},\mathbf{w}_{m,k}^{c,r}\} ) = &  \bigr| \sum\nolimits_{\ell \in \mathcal{B}_i} (\mathbf{u}^{c,r}_{i})^\mathrm{H}\mathbf{H}_{\ell,i}^{c,r} \mathbf{w}_{\ell,i}^{c,r}  \bigr|^2 \notag \\
			& \!\!\!\!\!\!\!\!\!\! \!\!\!\!\!\!\!\!\!\! \!\!\!\!\!\!\!\!\!\! \!\!\!\!\!\!\!\!\!\! \!\!\!\!\!\!\!\!\!\!  \times\bigr(  \sum\nolimits_{\ell \in \mathcal{M}}\sum\nolimits_{\substack{t \in \mathcal{K}_\ell}} |(\mathbf{u}^{c,r}_{i})^\mathrm{H}b_{\ell,t}^{c,r} \mathbf{H}_{\ell,i}^{c,r} \mathbf{w}_{\ell,t}^{c,r} |^2   \\
			&\!\!\!\!\!\!\!\!\!\!  \!\!\!\!\!\!\!\!\!\!  \!\!\!\!\!\!\!\!\!\!  \!\!\!\!\!\!\!\!\!\! \!\!\!\!\!\!\!\!\!\!+\sum\nolimits_{j \in \mathcal{I}, j \ne i} \bigr|\sum\nolimits_{\ell \in \mathcal{B}_{j}} (\mathbf{u}^{c,r}_{i})^\mathrm{H}b_{\ell,j}^{c,r}\mathbf{H}_{\ell,i}^{c,r} \mathbf{w}_{\ell,j}^{c,r}   \bigr|^2 +\sigma^2  \bigr)^{-1}. \notag
		\end{align}
	\vspace{-0.4cm}
	
	\noindent Thus, the data rate for NJT-UE $k$ in cell $m$ on RBG $(c,r)$ is
	
	\vspace{-0.4cm}
		\begin{align}\label{realrate}
			f_{m, k}^{c,r}(\{ b_{m,k}^{c,r},\mathbf{w}_{m,k}^{c,r}\} ) = b_{m, k}^{c,r}\log(1 + \gamma_{m, k}^{c,r}),
		\end{align}
	\vspace{-0.4cm}
	
	\noindent and the rate for JT-UE $i$ in cell $m$ on RBG $(c,r)$ is
	
	\vspace{-0.35cm}
	
		\begin{align}\label{realrate11}
			f_{i}^{c,r}(\{ b_{m,k}^{c,r},\mathbf{w}_{m,k}^{c,r}\} ) =  b_{m,i}^{c,r}\log(1 + \gamma_{i}^{c,r}).
		\end{align}
	 \vspace{-0.5cm}
	
	\noindent Note that the scheduling indicators for JT-UE $i$ need to be consistent across all its serving O-RUs, i.e., $b_{m,i}^{c,r} = b_{n,i}^{c,r},$ $\forall m,n \in \mathcal{B}_i$, to ensure coherent joint transmission. 
	\vspace{-0.3cm}
	\subsection{Transceiver Design Based on EZF}

    Given limited computational resources and real-time processing requirements, (semi-)closed-form BF solutions are generally preferred over iterative methods in practice. This motivates us to formulate the scheduling problem using EZF-BF \cite{sun2010eigen}. Unlike ZF-BF that directly inverts the channel matrix and can be sensitive to poor conditioning, EZF-BF operates in the dominant channel eigenspace obtained via eigen-decomposition. By suppressing interference across the most significant eigenmodes, EZF-BF improves numerical stability and robustness in ill-conditioned channel scenarios \cite{kaziu2024approximate}. 
	
	Specifically, for each NJT-UE $ k \in \mathcal{K}_m$ scheduled on RBG $(c,r)$, with the intra-cell CSI $\mathbf{H}_{m, k}^{c,r}$, the serving O-RU $m \in \mathcal{M}$ first conducts singular value decomposition (SVD) as
	$\mathbf{H}_{m, k}^{c,r} =\mathbf{T}_{k}^{c,r} \mathbf{\Lambda}_{ k}^{c,r} ({\mathbf{V}_{m, k}^{c,r}})^\mathrm{H}$, 
	where $\mathbf{\Lambda}_{k}^{c,r}$ is an $N_r \times N_t$ diagonal matrix containing singular values of $\mathbf{H}_{m,k}^{c,r}$ in the decreasing order along its main diagonal, the largest of which is denoted by $\lambda_{k}^{c,r}$. The terms $\mathbf{T}_{k}^{c,r} \in \mathbb{C}^{N_r \times N_r}$ and $\mathbf{V}_{m,k}^{c,r} \in \mathbb{C}^{N_t \times N_t}$ are the corresponding unitary matrix composed of the left singular vectors and right singular vectors, respectively.  
	
	Similarly, for JT-UE $i  \in \mathcal{U}_m$ served by multiple O-RUs, the transceiver is jointly designed based on all channels from its serving O-RUs.  
	Specifically, by stacking all the channel matrices from the serving O-RUs $m \in \mathcal{B}_i$ to JT-UE $i$, its aggregated channel matrix on RBG $(c,r)$ is 
	
	\vspace{-0.4cm}
	
		\begin{align} \label{eq:SVd_JT}
			\mathbf{H}_{i}^{c,r} \triangleq   [ \mathbf{H}_{m,i}^{c,r} ]_{m \in \mathcal{B}_i}  \in \mathbb{C}^{N_r \times |\mathcal{B}_i|N_t }. 
		\end{align}
	\vspace{-0.5cm}
	
	\noindent Let its SVD be expressed as $\mathbf{H}_{i}^{c,r} = \mathbf{T}_{i}^{c,r} \mathbf{\Sigma}_{i}^{c,r} (\mathbf{V}_{i}^{c,r})^\mathrm{H},$
	where $\mathbf{V}_{i}^{c,r}\in \mathbb{C}^{|\mathcal{B}_i|N_t \times |\mathcal{B}_i|N_t}$ is the unitary right singular vector matrix. $\mathbf{\Sigma}_{i}^{c,r} \in \mathbb{C}^{N_r \times |\mathcal{B}_i|N_t}$ holds the singular values in descending order, with its largest $\lambda_{i}^{c,r}$. 
	It can be found from \eqref{eq:SVd_JT} that $\{\mathbf{H}_{m,i}^{c,r}\}_{m \in \mathcal{B}_i}$ share the same left singular vectors and singular values, i.e., $
	\mathbf{H}_{m,i}^{c,r} \triangleq \mathbf{T}_i^{c,r} \mathbf{\Sigma}_i^{c,r} (\mathbf{V}_{m,i}^{c,r})^\mathrm{H}$, where $\mathbf{V}_{m,i}^{c,r}$ is the corresponding sub-matrix that can be attained from $\Vb_i^{c,r}$.
	

    While it is possible to consider global EZF \cite{li2014multicell}, which jointly nulls both intra- and inter-cell interference, it may not be a good choice in practice. First,  its matrix inversion complexity scales with the total number of scheudled users, and with a limited number of transmit antennas, the available spatial degrees of freedom may be insufficient for effective global nulling. Moreover, since cell-center users experience negligible ICI, enforcing global nulling can waste spatial resources. Therefore, we adopt intra-cell EZF in this work, where each cell only nulls the intra-cell interference among its scheduled UEs. 
    
    To illustrate this, we consider both NJT and JT UEs $\mathcal{A}_m$ in each cell $m$. Let $\mathbf{t}_{k}^{c,r}$ and $\mathbf{v}_{m,k}^{c,r}$ be the first column of $\mathbf{T}_{k}^{c,r}$ and $\mathbf{V}_{m,k}^{c,r}$, respectively. To maximize the receive SNR of UE $k \in \Ac_m$ in RBG $(c,r)$, we set $\mathbf{u}_{k}^{c,r} = \mathbf{t}_{k}^{c,r}$ and it leads to
	
	\vspace{-0.4cm}
	
		\begin{align}\label{uH}   (\mathbf{u}_{k}^{c,r})^{\mathrm{H}}\mathbf{H}_{m,k}^{c,r} = \lambda_k^{c,r} (\mathbf{v}_{m,k}^{c,r})^{\mathrm{H}}.
		\end{align}
	\vspace{-0.5cm}
	
	\noindent Based on the equivalent channel in \eqref{uH}, EZF projects intended signals onto the orthogonal complement of its interference channel subspace, thereby mitigating inter-UE interference. Specifically, denote the set of UEs in cell $m$ that are scheduled over RBG $(c,r)$ as $\mathcal{A}_m^{c,r} \triangleq \{a_1 ,..., a_m\} \in \mathcal{A}_m$, whose cardinality is given by $ |\mathcal{A}_m^{c,r}|= \sum\nolimits_{t\in \Ac_m}b^{c,r}_{m,t}$. By stacking the right singular vectors of these UEs as
	
	\vspace{-0.2cm}
	
		\begin{equation}\label{eqV}
			\hat{\mathbf{V}}_m^{c,r} \triangleq \left[ {\mathbf{v}}_{m,a_1}^{c,r},  \ldots, {\mathbf{v}}_{m,a_{m}}^{c,r} \right], 
		\end{equation}
	\vspace{-0.5cm}
	
	\noindent the EZF TBF at O-RU $m$ over RBG $(c,r)$ is computed by
	
	\vspace{-0.35cm}
	
		\begin{align}\label{bF}
			\!\! \hat{\mathbf{W}}_m^{c,r} = &\hat{\mathbf{V}}_m^{c,r} ( (\hat{\mathbf{V}}_m ^{c,r})^\mathrm{H} \hat{\mathbf{V}}_m^{c,r} )^{-1} =\left[ \hat{\mathbf{w}}_{m,a_1}^{c,r}, \ldots, \hat{\mathbf{w}}_{m,a_{m}}^{c,r} \right].
		\end{align} 
	\vspace{-0.5cm}
	
	\noindent Given that the total transmit power per RBG $(c,r)$ in each O-RU $m$ as $P$, each EZF-TBF with equal power allocation (EPA) is 
	
	\vspace{-0.2cm}
	
		\begin{equation}\label{normalized}
			\mathbf{w}_{m,j}^{c,r} = \sqrt{P/|\mathcal{A}_m^{c,r}|}\hat{\mathbf{w}}_{m,j}^{c,r}/\| \hat{\mathbf{w}}_{m,j}^{c,r} \|.
		\end{equation}
	 \vspace{-0.5cm}
	
	\noindent 
    Notice that the structure of the spatial direction matrix $\mathbf{\hat{V}}^{c,r}_m$ in \eqref{eqV} and the unnormalized TBF matrix $\mathbf{\hat{W}}^{c,r}_m$ in \eqref{bF} lead to the orthogonality $
	(\mathbf{\hat{V}}^{c,r}_{m})^\mathrm{H}  \mathbf{\hat{W}}^{c,r}_{m} = \mathbf{I}, 
	\forall m,c,r.$
	Accordingly, for any two UEs $j$ and $k$ in cell $m$, $\mathbf{w}^{c,r}_{m,j}$ in \eqref{normalized} satisfies
	
	\vspace{-0.35cm}
	
		\begin{align}\label{vw}
			\!\!\!(\mathbf{v}_{m,k}^{c,r})^\mathrm{H}\mathbf{w}_{m,j}^{c,r} \!= \!\begin{cases} {\!\!\sqrt{P/|\mathcal{A}_m^{c,r}|}}/{\|\mathbf{\hat{w}}_{m,k}^{c,r}\|}, & \text{if $j\!=\!k$},  \\ \!\! 0, & \text{otherwise.}\end{cases}
		\end{align}
	\vspace{-0.7cm} 
	
	\subsection{Multi-Dimensional User Scheduling Problem} 
	Building upon the EZF transceiver design, we develop a scheduling framework that supports a multi-service coexistence scenario with two types of UEs:
	\begin{itemize}
		\item \textbf{UEs with QoS requirement}: These UEs require low-delay guarantees—treated as a key QoS metric—such as image delivery for interactive applications. Generally, they have a finite data volume at the O-DU.
		\item \textbf{UEs without QoS requirement}: These UEs do not have stringent delay constraints and primarily aim to maximize sustained throughput, such as high-definition video downloads. Their traffic is typically characterized by continuous data generation at the O-DU.
	\end{itemize}
	In this work, our design aims to maximize the sum rate of the second type of UEs $k \in \hat\Kc_m^Q, m \in \Mc, i \in \hat\Ic^Q$, while ensuring the QoS requirements $\{Q_k, Q_i\}$ of the first type of UEs $k \in \Kc_m^Q, m \in \Mc, i \in \Ic^Q$.
	Accordingly, the multi-dimensional user scheduling optimization can be formulated as the following nonlinear integer programming (NLIP) problem
	
	\vspace{-0.4cm}
	
		\begin{subequations}\label{SRmax-original}
			\begin{align}
				\!\!\!\!\max\limits_{\{b^{c,r}_{m,k}\}}  & \! \sum\nolimits_{(c,r)}\!\!\left(\sum\nolimits_{m \in \mathcal{M}}\! \sum\nolimits_{k \in \mathcal{\hat{K}}_m^Q} \!f_{m, k}^{c,r}\!+\!\sum\nolimits_{i \in\mathcal{\hat{I}}^Q} \! f_{i}^{c,r} \right)\label{SRmax0_a}\\
				\text { s.t. } & b_{m, k}^{c,r} \in \{0,1\}, \forall m, c,r, k, \label{Binary}\\
				&   b_{m,i}^{c,r} =  b_{n,i}^{c,r} , \forall  m,n\in \mathcal{B}_i ,  \forall i \in \mathcal{I}, \forall c,r,\label{JT0}\\
				& \sum\nolimits_{(c,r)} f_{m,k}^{c,r} \geq Q_{k}, \forall m, k  \in \mathcal{K}_m^{ {Q}},  \label{KNF0}\\
				&\sum\nolimits_{(c,r)} f_{i}^{c,r} \geq Q_i, \forall i \in \mathcal{I}^{ {Q}},\label{INF0}
			\end{align}
		\end{subequations}
	\vspace{-0.4cm}
	
	\noindent where the rate function $\{f^{c,r}_{m,k}, f^{c,r}_{i}\}$ is derived from \eqref{realrate} and \eqref{realrate11}. In addition, constraint \eqref{JT0} enforces scheduling consistency for JT-UEs across all coordinating O-RUs per RBG. Meanwhile, constraints \eqref{KNF0} and \eqref{INF0} specify QoS requirements for NJT-UEs and JT-UEs, respectively.  
	
	In problem \eqref{SRmax-original}, the TBF and the scheduling spans over multiple dimensions, including cells, CCs, JT-UEs, and NJT-UEs with or without QoS constraints. This makes \eqref{SRmax-original} far more intricate than conventional formulations \cite{chen2023om}, which are limited to parts of these dimensions. Optimization of \eqref{SRmax-original} is quite challenging due to the following factors:
	
	\begin{enumerate}
		\item \textbf{Combinatorial Explosion}: As an NLIP problem, its solution space grows exponentially as $\mathcal{O}(2^{K \times C \times R})$ with increasing UEs ($K$) and RBGs ($C \times R$), rendering exhaustive search computationally prohibitive;
		
		\item \textbf{Strong Coupling}: The generation of EZF beamformers $\{\wb_{m,j}^{c,r}\}$ depends on the scheduling decisions $\{b^{c,r}_{m,k}\}$. Accordingly, the scheduling variables and TBF variables are strongly coupled across multiple dimensions;
		
		\item \textbf{NP-hardness}: The non-convexity and the non-smooth structure of problem~\eqref{SRmax-original} make it NP-hard~\cite{denis2021improving}.
	\end{enumerate}
	Even without considering distributed optimization tailored to the O-RAN architecture, these characteristics render the design of a polynomial-complexity algorithm for achieving a local optimum highly non-trivial. To tackle them, we first propose a novel reformulation scheme to simplify the problem structure. 
	\vspace{-0.4cm} 
	\section{A Novel Problem Reformulation} \label{sec:reformu}    
	In this section, we propose a novel reformulation to simplify problem \eqref{SRmax-original} by leveraging the system properties. Notice that the problem complexity primarily stems from the rate expressions $\{f_{m,k}^{c,r}, f_i^{c,r}\}$. To handle them, our proposed reformulation scheme comprises two phases:  \textbf{Phase 1}
	leverages interference characteristics and the EZF transceiver structure to simplify the SINR formulation. Using the log-fractional form of the rate and Jensen’s inequality, a closed-form rate approximation for JT-UEs is derived, which decouples variables across PUs and reduces coordination complexity. \textbf{Phase 2} applies massive MIMO asymptotic analysis \cite{vershynin2009high} to eliminate the implicit dependence on EZF TBF in \eqref{realrate} and \eqref{realrate11}, achieving tractable rate expressions.  
	\vspace{-0.3cm}
	\subsection{ \textbf{Phase 1}}
	Recall that transmission strategies differ between cell-center and cell-edge UEs. For cell-center UEs, the desired signal from the serving O-RU typically dominates ICI, as the serving link’s channel gain is much stronger than others. Hence, NJT is sufficient for these UEs. In contrast, cell-edge UEs experience comparable signal strengths from multiple O-RUs, necessitating JT to mitigate interference. Accordingly, it is reasonable to make the following assumption \cite{chen2022m,seifi2011coordinated}. 
	
	\noindent \textbf{Assumption 1}: The ICI from the non-serving O-RUs to NJT-UEs is negligible, while JT-UEs only experience interference from their served O-RUs.

	\begin{figure}[!htbp] 
		\centering	        {\includegraphics[width=0.28\textwidth]{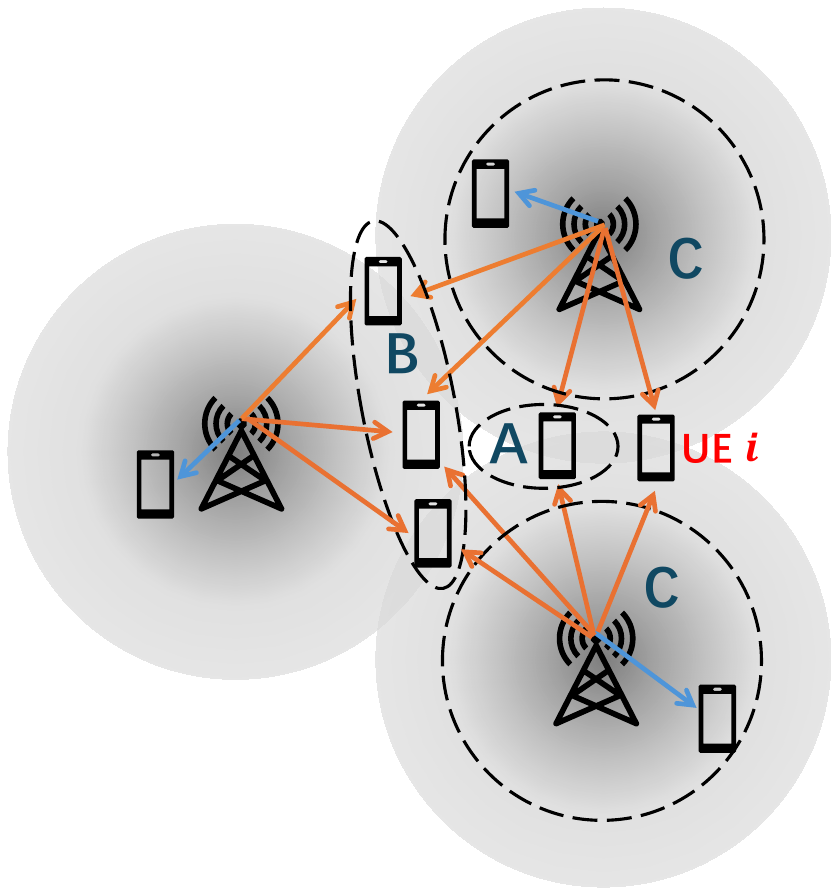}}	
		\vspace{-0.15cm}
		\caption{{The interference sources {(circled by dashed lines)} experienced by JT-UE $i$.}} 
		\label{interference} 
	\end{figure}  
	
	\noindent Based on this assumption, the ICI term in the SINR \eqref{gamma} of NJT-UE $k\in \Kc_m$ can be omitted, yielding an approximate SINR expression \eqref{gamma1} shown at the top of the next page.
	\begin{figure*}[!t]

			\begin{align}\label{gamma1}
				\!\!\!\gamma_{m,k}^{c,r} & \approx |   (\mathbf{u}^{c,r}_{k})^\mathrm{H} \mathbf{H}_{m,k}^{c,r} \mathbf{w}_{m,k}^{c,r}  |^2 \bigr( \sum\nolimits_{\substack{j \neq k, j \in \mathcal{K}_m }} |(\mathbf{u}^{c,r}_{k})^\mathrm{H} b_{m,j}^{c,r} \mathbf{H}_{m,k}^{c,r} \mathbf{w}_{m,j}^{c,r}  |^2  + \sum\nolimits_{i \in \mathcal{U}_m} | (\mathbf{u}^{c,r}_{k})^\mathrm{H} b_{m,i}^{c,r} \mathbf{H}_{m,k}^{c,r} \mathbf{w}_{m,i}^{c,r}  |^2 + \sigma^2 \bigr)^{-1}  
			\end{align}
			\vspace{-0.4cm}
			\begin{align}\label{gamma2}
				&\gamma_{i}^{c,r} \approx    \bigr|  \sum\nolimits_{m \in \mathcal{B}_i}(\mathbf{u}^{c,r}_{i})^\mathrm{H}  \mathbf{H}_{m,i}^{c,r} \mathbf{w}_{m,i}^{c,r}   \bigr|^2   \bigr( \sum\nolimits_{j\in (\cap_{m\in \mathcal{B}_i}\mathcal{U}_m), j \ne i} \bigr|\sum\nolimits_{m \in \mathcal{B}_i} (\mathbf{u}^{c,r}_{i})^\mathrm{H}b_{m,j}^{c,r}\mathbf{H}_{m,i}^{c,r} \mathbf{w}_{m,j}^{c,r}   \bigr|^2     +\\
				&\sum\nolimits_{\substack{j \notin (\cap_{m\in \mathcal{B}_i}\mathcal{U}_m), j \in (\cup_{m\in \mathcal{B}_i} \mathcal{U}_m) }} \bigr|\sum\nolimits_{m\in (\mathcal{B}_i\cap \mathcal{B}_j) }(\mathbf{u}^{c,r}_{i})^\mathrm{H}b_{m,j}^{c,r}\mathbf{H}_{m,i}^{c,r} \mathbf{w}_{m,j}^{c,r}   \bigr|^2 \!+\!\sum\nolimits_{\substack{m \in \mathcal{B}_i}}\sum\nolimits_{j \in \mathcal{K}_m} |(\mathbf{u}^{c,r}_{i})^\mathrm{H}b_{m,j}^{c,r} \mathbf{H}_{m,i}^{c,r} \mathbf{w}_{m,j}^{c,r}   |^2 +  \sigma^2  \bigr)^{-1} \notag
			\end{align}
		
		\vspace{-0.3cm}
		\hrulefill
	\end{figure*}
	In \eqref{gamma1}, the first and second terms in the denominator represent the intra-cell interference caused to user $k$ by the NJT-UEs and JT-UEs jointly scheduled on RBG $(c,r)$ from cell $m$, respectively. 
	
	For JT-UEs $i$, the interference is more complex and originates from three types of UEs (see Fig.~\ref{interference}):
	\begin{enumerate}
		\item \textit{Co-scheduled JT-UEs} with identical serving O-RU set ($\cap_{m\in \mathcal{B}_i}\mathcal{U}_m$, area A);
		\item \textit{Partially overlapping JT-UEs} sharing subsets of $\mathcal{B}_i$ ($\cup_{m\in \mathcal{B}_i}\mathcal{U}_m \setminus \cap_{m\in \mathcal{B}_i} \mathcal{U}_m$, area B);
		\item \textit{NJT-UEs} served by $\mathcal{B}_i$'s O-RUs (area C).
	\end{enumerate} 
	Accounting for these interference sources, the SINR for JT-UE $i \in \Ic$ in \eqref{gamma2} is shown at the top of next page.
	
	The denominators in \eqref{gamma1} (for NJT-UEs) and \eqref{gamma2} (for JT-UEs) reveal that every term of inter-UE interference suffered by one UE stems from its serving O-RUs. 
	Based on the properties in \eqref{uH} and \eqref{vw}, the SINR expression for NJT-UE $k\in \Kc_m$ in  \eqref{gamma1} can be simplified as
	
	\vspace{-0.2cm}
	
		\begin{align}\label{SINR_BF}
			\gamma_{m,k}^{c,r} = &    {(\lambda_{k}^{c,r})^2P}/{(|\mathcal{A}_m^{c,r}| \cdot \| \hat{\mathbf{w}}_{m,k}^{c,r} \|^2\sigma^2)},
		\end{align}
	\vspace{-0.2cm}
	
	\noindent where $\|\hat{\mathbf{w}}_{m,k}^{c,r} \|^2$ depends on the scheduling variables of O-RU $m$ implicitly, as shown in \eqref{eqV} and \eqref{bF}. 	
	Similarly, the SINR for JT-UE $i$ on RBG $(c,r)$ in \eqref{gamma2} can be further expressed as
	
	\vspace{-0.35cm}
	
		\begin{align}\label{JT SINR}
			\gamma_{i}^{c,r} =    \big| \sum\nolimits_{m \in \mathcal{B}_i} \lambda_{i}^{c,r} \sqrt{P}/ (\sqrt{|\mathcal{A}_m^{c,r}|}\cdot { \|\hat{\mathbf{w}}_{m,i}^{c,r}\| \sigma} )\big|^2.
		\end{align}
	\vspace{-0.35cm}

    
	{Since RBGs in MU-MIMO systems are typically allocated to UEs with good channel quality and low inter-UE interference, the resulting SINR on scheduled RBGs is generally high \cite{Yubo_TWC24}.} Therefore, the rate expression can be approximated by
	$\log(1+\gamma)\!\approx\!\log(\gamma)$, enabling the rates of both NJT-UE $k$ and JT-UE $i$ on RBG $(c,r)$ to be respectively expressed as
	
	\vspace{-0.4cm}
	
		\begin{subequations}
			\begin{align}
				\label{rate_appro}
				\tilde{f}_{m, k}^{c,r} &\approx  b_{m,k}^{c,r}\log(\gamma_{m,k}^{c,r}), k \in \Kc_m,\\
				\tilde{f}_{i}^{c,r} &\approx  b_{m,i}^{c,r}\log(\gamma_{i}^{c,r}), i \in \Ic.\label{rate_appro1}
			\end{align}
		\end{subequations}
	\vspace{-0.4cm}
	
	However, unlike the NJT-UE rate in \eqref{rate_appro}, the JT-UE rate \eqref{rate_appro1} exhibits coupling across multiple O-RUs, posing a challenge for distributed scheduling. To overcome this, we construct a separable lower bound for \(\tilde{f}_{i}^{c,r} \) using the decomposition technique in \cite{bian2024decentralizing}. 
	Specifically, we have  
	
	\vspace{-0.4cm}
	        \begin{subequations}\label{eq:gamma_JT_s}
			\begin{align} 
				 &\tilde{f}^{c,r}_{i}
				\stackrel{(a)}{\geq} \ b_{m,i}^{c,r}\log\left(\sum\nolimits_{m \in \mathcal{B}_i}{\gamma}_{m,i}^{c,r}\right),\\
				\stackrel{(b)}{\geq} &\ \sum\nolimits_{m\in \mathcal{B}_i}{b_{m,i}^{c,r}\log\left(|\mathcal{B}_i|\gamma_{m,i}^{c,r}\right)}/{|\mathcal{B}_i|}\triangleq \sum\nolimits_{m\in \mathcal{B}_i}\tilde{f}_{m,i}^{c,r}, \label{JT_rate lb}
			\end{align}
		\end{subequations}

	\noindent where each \( \gamma_{m,i}^{c,r} \) shares the same structural form as \eqref{SINR_BF}. Inequality (a) holds by the Cauchy–Schwarz inequality since all terms in the numerator of \eqref{JT SINR} are positive, while (b) follows from Jensen’s inequality.

	Comparing \eqref{JT_rate lb} with \eqref{realrate11} for JT-UEs (and \eqref{rate_appro} with \eqref{realrate} for NJT-UEs), one can observe that the coupling in the original rate expressions has been effectively alleviated.
    To streamline subsequent analysis, we unify the rate expressions of NJT-UEs in \eqref{rate_appro} and sub-terms within the rate expressions of JT-UEs in  \eqref{JT_rate lb} into a compact form as follows 
	
	\vspace{-0.4cm}
	
		\begin{align}\label{rate_approx1}
			\tilde{f}_{m,t}^{c,r}= \begin{cases}b_{m,t}^{c,r}\log(\gamma_{m,t}^{c,r}), & \text{if UE $t\in\Kc_m$,} \\ {b_{m,t}^{c,r}\log\left(|\mathcal{B}_t|\gamma_{m,t}^{c,r}\right)}/{|\mathcal{B}_t|}, & \text{if UE $t\in\Uc_m$.}\end{cases}
		\end{align}
	\vspace{-0.5cm}
	
	\subsection{\bf{Phase 2}}
	So far, the rate expression $\{\tilde{f}_{m, t}^{c,r}\}$ appears to depend only on scheduling variables and TBFs of O-RU $m$ over RBG $(c,r)$. However, the TBFs in \eqref{SINR_BF} are still implicitly coupled with the scheduling decisions within the cell, as their structure depends on the set of concurrently scheduled users. Moreover, ${\mathbf{\hat{w}}^{c,r}_{m,t}}$ exists in a fractional structure as in \eqref{bF}, making the optimization of problem \eqref{SRmax-original} still challenging.  
	To address these challenges, it is essential to characterize how scheduling decisions influence the achievable rates through the TBF structure. Toward this goal, we establish the following theorem by leveraging massive MIMO asymptotics \cite{vershynin2009high} and the structure of EZF-TBF. 
	
	\vspace{-0.1cm}
	\begin{Theorem} \label{Th:1}
		When the number of transmit antennas is sufficiently large, the rate for NJT-UE $k$ served by cell $m$ associated with RBG $(c,r)$ converges to
		
		\vspace{-0.3cm}
		
			\begin{align}\label{JTSINR1}
				\tilde{f}_{m,k}^{c,r} \approx b_{m,k}^{c,r}(\psi_{m, k}^{c,r} +\sum\nolimits_{j \in \mathcal{A}_m\setminus\{k\}} \! b_{ m,j}^{c,r} d_{m, j, k}^{c,r} - g_m^{c,r} ),  
			\end{align} 
		\vspace{-0.3cm}
		
		\noindent where ${\psi}_{m,k}^{c,r}  \triangleq       \log({{({\lambda}}_{k}^{c,r})^2|\mathbf{v}_{m,k}^{c,r} |^2 P}/{\sigma^2})$, ${d}_{m,j,k}^{c,r} \triangleq  \log( 1 -  \eta_{m,j,k}^{c,r})$, $\eta_{m,j,k}^{c,r} \triangleq { | (\mathbf{v}_{m,j}^{c,r}) ^\mathrm{H} \mathbf{v}_{m,k}^{c,r} |^2}/({|\mathbf{v}_{m,j}^{c,r}|^2|\mathbf{v}_{m,k}^{c,r}|^2})$, $g_m^{c,r} \triangleq \log(\phi_m^{c,r}) $ and $\phi_m^{c,r} \triangleq \sum\nolimits_{t\in \Ac_m}b^{c,r}_{m,t} $.
	\end{Theorem}
	\begin{proof}
	As the SINR derivations across different RBGs follow the same procedure, we drop the superscript $(c,r)$ in the following for notational simplicity.
    
	First, from the EZF-TBF construction in \eqref{bF}, it holds that $\hat{\mathbf{W}}_m^\mathrm{H} \hat{\mathbf{W}}_m = (  \hat{\mathbf{V}}_m ^\mathrm{H} \hat{\mathbf{V}}_m )^{-1}$. W.l.o.g., assume that the TBF $\hat \wb_{m,k}$ of UE $k$ occipies the $j$-th column of $\hat{\mathbf{W}}_m$ in \eqref{bF}, then
	
	\vspace{-0.2cm}
	
		\begin{equation}\label{eq:ww}
			\|\hat{\mathbf{w}}_{m,k} \|_2^2 = (  \hat{\mathbf{V}}_m ^\mathrm{H} \hat{\mathbf{V}}_m)_{j,j}^{-1}.
		\end{equation}
	\vspace{-0.5cm}
	
	\noindent Denote  $\mathbf{\hat{V}}_{m,\bar{j}}$ as the sub-matrix of $\mathbf{\hat{V}}_m$ without the $j$-th column $\mathbf{v}_{m,k}$. Applying the block matrix inversion formula \cite{li2015multicell} yields 

    \vspace{-0.3cm}
	
		\begin{align} \label{eq:invV}
			\!{ (\hat{\mathbf{V}}_m ^\mathrm{H} \hat{\mathbf{V}}_m)_{j,j}^{-1}} \!=\! \frac{1}{ \mathbf{v}_{m,k} ^\mathrm{H}(\mathbf{I}-  \hat{\mathbf{V}}_{m,\bar{j}} ( \hat{\mathbf{V}}_{m,\bar{j}}^\mathrm{H} \hat{\mathbf{V}}_{m,\bar{j}} )^{-1} \hat{\mathbf{V}}_{m,\bar{j}} ^\mathrm{H}) \mathbf{v}_{m,k}}. 
		\end{align}
	\vspace{-0.3cm}
    
	%
	
	\noindent Combining \eqref{eq:invV} with \eqref{eq:ww} leads to
	
	\vspace{-0.3cm}
				\begin{align}\label{eq:inv_ww}
			\!\!{1}/{\|\hat{\mathbf{w}}_{m,k} \|_2^2} =   \mathbf{v}_{m,k} ^\mathrm{H}(\mathbf{I}-  \hat{\mathbf{V}}_{m,\bar{j}} ( \hat{\mathbf{V}}_{m,\bar{j}}^\mathrm{H} \hat{\mathbf{V}}_{m,\bar{j}} )^{-1} \hat{\mathbf{V}}_{m,\bar{j}} ^\mathrm{H}) \mathbf{v}_{m,k}.
		\end{align}
	\vspace{-0.45cm}
	
	\noindent Further, it follows from \eqref{eqV} that the diagonal elements of $\hat{\mathbf{V}}_{m,\bar{j}}^{\mathrm{H}} \hat{\mathbf{V}}_{m,\bar{j}}$ are given by $\mathbf{v}_{m,z}^{\mathrm{H}}\mathbf{v}_{m,z},\, z \neq k$, while its off-diagonal elements represent the inner products between the right singular vectors associated with different UEs on the same RBG. According to the asymptotic properties of massive MIMO
    \cite{vershynin2009high}, as the number of transmit antennas $N_t$
    approaches infinity, the orthogonality condition $\mathbf{v}_{m,z}^{\mathrm{H}}\mathbf{v}_{m,k} \rightarrow 0$
    holds for any pair of users $z \neq k$ with different channels. Hence,
    $\hat{\mathbf{V}}_{m,\bar{j}}^{\mathrm{H}} \hat{\mathbf{V}}_{m,\bar{j}}$
    is asymptotically diagonal for large $N_t$, allowing
    \eqref{eq:inv_ww} to  be well approximated as
    
	
	\vspace{-0.3cm}
	
		\begin{align} \label{eq:app_diag}
			\!\!\!|\mathbf{v}_{m,k}|^2\bigr(1 \! - \!\!\sum\nolimits_{\substack{z \in \mathcal{A}_m, z \neq k}}  \!\!{b_{m,z} | \mathbf{v}_{m,z} ^\mathrm{H} \mathbf{v}_{m,k} |^2}\!/({|\mathbf{v}_{m,z}|^2|\mathbf{v}_{m,k}|^2})\bigr).
		\end{align}
	\vspace{-0.3cm}
	
	\noindent Moreover, since the approximation $1-\sum_{i=1}^n x_i \approx \prod_{i=1}^n\left(1-x_i\right)$ holds when $x_i \approx 0$, we have
	
	\vspace{-0.4cm}
	
		\begin{align} \label{eq:prod}
			&1  - \sum\nolimits_{ {z \in \mathcal{A}_m, z \neq k}}  {b_{m,z} | \mathbf{v}_{m,z} ^\mathrm{H} \mathbf{v}_{m,k} |^2}/({|\mathbf{v}_{m,z}|^2|\mathbf{v}_{m,k}|^2})  \\
			&\approx \prod\nolimits_{ {z \in \mathcal{A}_m, z \neq k}} \left( 1 -  {b_{m,z} | \mathbf{v}_{m,z} ^\mathrm{H} \mathbf{v}_{m,k} |^2}/({|\mathbf{v}_{m,z}|^2|\mathbf{v}_{m,k}|^2} )\right)  \notag.
		\end{align}
	\vspace{-0.4cm}
	
	\noindent Substituting \eqref{eq:prod} and \eqref{eq:app_diag}  into \eqref{SINR_BF} obtains the \eqref{JTSINR1}.
	\end{proof}\vspace{-0.2cm}
	
	\noindent Compared to \eqref{rate_appro}, \eqref{JTSINR1} successfully eliminates the intermediate TBF variables and depends solely on the scheduling variables within serving cell $m$. Moreover, it reveals the interplay among different UEs' scheduling decisions. From \eqref{JTSINR1}, we observe that the rate is influenced by the following three components:
	\begin{enumerate}
		\item ${\psi}_{m,k}^{c,r}$, corresponds to the rate contribution from the SNR of UE $k$ when scheduled exclusively on RBG $(c,r)$;
		
		\item $\sum\nolimits_{j \in \mathcal{A}_m\setminus\{k\}} \! b_{ m,j}^{c,r} d_{m, j, k}^{c,r}$, where $d_{m, j, k}^{c,r}$ is negative, quantifies the loss due to interference from other UEs with correlated channels (i.e., co-channel interference);
		
		\item $g_{m}^{c,r}$, captures the impact of power sharing among co-scheduled UEs.
	\end{enumerate}
    Analogously, by applying Theorem \ref{Th:1} to the first term in \eqref{JT_rate lb}, we have for the JT-UE $i \in \mathcal{U}_m$ that 
	
	\vspace{-0.25cm}
	
		\begin{align} \label{rate1}
			\!\!\tilde{f}_{m,i}^{c,r} \approx b_{m,i}^{c,r}(\tilde{\psi}_{m, i}^{c,r}\! +\!\sum\nolimits_{j \in \mathcal{A}_m\setminus\{i\}} \! b_{m,j}^{c,r} d_{m, j, i}^{c,r}\! - \!g_m^{c,r} )/|\mathcal{B}_i|,  
		\end{align}
	\vspace{-0.4cm}
	
	\noindent where $\tilde{\psi}_{m,i}^{c,r} = \log(|\mathcal{B}_i|)+{\psi}_{m,i}^{c,r}$.
	
	Following the above two-phase reformulation, one can observe that the original multiplicative structures in \eqref{gamma} and \eqref{gamma_JT} are transformed into a more tractable summation form.
	
	\vspace{-0.1cm}
	\section{Proposed Centralized Algorithm}
	\vspace{-0.05cm}
		In this section, we first develop a centralized algorithm for problem \eqref{SRmax-original} as a benchmark. By applying our proposed reformulation scheme in \eqref{JT_rate lb}, \eqref{SRmax-original} can be approximated to 
	
	\vspace{-0.3cm}
	
		\begin{subequations}\label{SRmax2}
			\begin{align}
				\max_{\{b^{c,r}_{m,k}\}}  & \sum_{(c,r)}\sum_{m \in \mathcal{M}}\left( \sum\nolimits_{k \in \mathcal{\hat{K}}_m^{Q}} \tilde{f}_{m, k}^{c,r}+\sum\nolimits_{i \in \mathcal{\hat{U}}_m^Q}  \tilde{f}_{m,i}^{c,r}\right) \label{SRmax1_a}\\
				\text { s.t. } & b_{m, k}^{c,r} \in \{0,1\}, \forall m, c,r, k, \label{Binary1}\\
				&   b_{m,i}^{c,r} =  b_{n,i}^{c,r} , \forall  m,n\in \mathcal{B}_i ,  \forall i \in \mathcal{I}, \forall c,r,\label{JT1}\\
				& \sum\nolimits_{(c,r)} \tilde{f}_{m, k}^{c,r} \geq Q_{k},  \quad \forall m, k  \in \mathcal{K}_m^{ {Q}}, \label{KNF}\\
				&\sum\nolimits_{(c,r)} \sum\nolimits_{m\in \mathcal{B}_i}\tilde{f}_{m,i}^{c,r} \geq Q_i, \quad \forall i \in \mathcal{I}^{ {Q}}. \label{INF}
			\end{align}
		\end{subequations}

	\noindent Compared with \eqref{SRmax-original}, problem \eqref{SRmax2} no longer has the intermediate TBF variables.  
	However, as shown in \eqref{JTSINR1} and \eqref{rate1}, the rate expression $\{\tilde{f}_{m,t}^{c,r}\}$ in \eqref{SRmax1_a} remains non-convex, while the non-convex constraints in  \eqref{JT1}-\eqref{INF} further exacerbate the optimization challenge. To address this, we first employ a \textit{min}-based penalty function \cite{nocedal1999numerical} to the rate constraints \eqref{KNF} and \eqref{INF}, and reformulate problem \eqref{SRmax2} as
	
	\vspace{-0.35cm}
	
    \begin{subequations}\label{gloabl}
        \begin{align} \label{minpenalty}
			\max\nolimits_{\{b^{c,r}_{m,k}\}}\ &G\\
				\text { s.t. } & b_{m, k}^{c,r} \in \{0,1\}, \forall m, c,r, k, \label{Binary2}\\
				&   b_{m,i}^{c,r} =  b_{n,i}^{c,r} , \forall  m,n\in \mathcal{B}_i ,  \forall i \in \mathcal{I}, \forall c,r,\label{JT2}
		\end{align}
    \end{subequations}

	\noindent where
	
	\vspace{-0.3cm}
	
		\begin{align} \label{G}
			&G  \triangleq \!\sum_{m \in \mathcal{M}}\bigg( \! \sum_{k \in \mathcal{\hat{K}}_m^Q} \!\sum_{(c,r)} \! \tilde{f}_{m, k}^{c,r} \! + \! \rho \! \sum_{k \in \mathcal{K}_m^{Q}} \!\!\min \!\bigg\{ \sum_{(c,r)} \tilde{f}_{m, k}^{c,r}, Q_{ k}\bigg\}\!\bigg) \notag\\
			&\!\! + \!\sum_{i \in\mathcal{\hat{I}}^Q} \!\sum_{m \in \mathcal{B}_i} \! \sum_{(c,r)} \! \tilde{f}_{m, i}^{c,r}  \! + \! \rho \! \sum_{i \in \mathcal{I}^{Q}} \! \min \bigg\{\sum_{m \in \mathcal{B}_i} \! \sum_{(c,r)} \tilde{f}_{m, i}^{c,r}, Q_i\bigg\}, 
		\end{align}

	\noindent and $\rho > 0$ is a penalty parameter. For the above transformation, the following theorem \cite{nocedal1999numerical} holds true.
	\vspace{-0.1cm}
	\begin{Theorem} 
		There exists an $\rho^*>0$, when $\rho>\rho^* $, the local optimal solution of the above problem \eqref{gloabl} is also the local optimal solution of the original constraint problem \eqref{SRmax2}. 
	\end{Theorem}
	
	Secondly, to maintain the scheduling consistency of JT-UEs, we introduce global variables $\{b^{c,r}_{i}\}$ into constraint \eqref{JT2} such that $b^{c,r}_{i}=b^{c,r}_{m,i},\forall m \in \mathcal{B}_i$. We propose to solve problem \eqref{gloabl} directly over binary variables in a BCD manner.
	Specifically, for each scheduling variable $b^{c,r}_{m,k}$ or $b^{c,r}_{i}$, when the remaining scheduling variables are fixed, the gain of scheduling a UE $t$ on RBG $(c,r)$ is defined as
	
	\vspace{-0.3cm}
	
		\begin{align} \label{eq:gain1}
			\!\!\!\!\text{gain}[c,r,t] \!\!=     \!\!\begin{cases}\!G(b^{c,r}_{m,t}=1) \!- \!G(b^{c,r}_{m,t}=0), & \!\!\!\!\!\!\text { if $t \in \mathcal{K}_m$}, \\ \!G(b^{c,r}_{t}=1) \!- \!G(b^{c,r}_{t}=0), & \!\!\!\!\!\!\text { if $t \in \mathcal{I}$.}\end{cases}
		\end{align}
	\vspace{-0.2cm}
	
	\noindent If the $\text{gain}[c,r,t]$ is positive, the corresponding scheduling variable will be set to 1; otherwise, it will be set to 0. By iteratively applying this process to all variables, we obtain a centralized scheduling algorithm as detailed in Algorithm \ref{algor1}. It is not difficult to verify that Algorithm \ref{algor1} will generate a non-decreasing sequence of objective values and converge to a local optimal solution \cite{jager2020blockwise}. 
    

	\begin{figure}[t]
		\vspace{-0.25cm}
		\begin{algorithm}[H]
				\caption{Proposed Centralized Scheduling Algorithm}
				\begin{algorithmic}[1] 
					\vspace{-0.2cm}\label{algor1}
					\STATE Initialize $\{b_{m, k}^{c,r}, b_{i}^{c,r}\}$; 
					\FOR{each UE $t$}
					\FOR{each RBG $(c,r)$}
					\IF{$\text{gain}[c,r,t] > 0$}
					\STATE $b_{m, t}^{c,r} = 1, t \in \mathcal{K}_m$ or $b_{t}^{c,r} = 1, t \in \mathcal{I}$; 
					\ELSE
					\STATE $b_{m, t}^{c,r} = 0, t \in \mathcal{K}_m$ or $b_{t}^{c,r} = 0, t \in \mathcal{I}$; %
					\ENDIF
					\ENDFOR
					\ENDFOR
					\STATE Repeat Steps 2-10 until maximum iterations reached. 
				\end{algorithmic}
		\end{algorithm}
		\vspace{-0.8cm}
	\end{figure}

	\vspace{-0.2cm}
	\section{Distributed Framework and Algorithm Design} \label{sec:decentralized}
	\vspace{-0.0cm}

    While effective, Algorithm 1 is a centralized scheme that can only be executed sequentially on a single computational node (e.g., a PU) and underutilizes the parallel processing resources of O-DU. As the network scales, this imposes a considerable computational burden on the single node. Therefore, it is crucial to design a distributed scheduling scheme that fully exploits the multi-PU and multi-core parallelism within the O-DU. This will not only reduce the computational load on each PU but also enable concurrent processing, thereby improving scheduling scalability in large-scale deployments.
	
	The key to distributed design is to offload the scheduling optimization to multiple cores in $\text{PU}_1 \sim \text{PU}_M$ to enable parallel processing, leaving only lightweight coordination to $\text{PU}_0$. For problem \eqref{SRmax2}, the QoS requirements in \eqref{KNF} need to be satisfied across CCs within one PU, and \eqref{INF} needs be jointly satisfied across PUs, which are strongly coupled. Besides, the scheduling consistency of JT-UE across different cells in \eqref{JT0} makes the distributed design more difficult. 	A natural candidate for handling such constraints typically employs \textit{consensus alternating direction method of multipliers} (ADMM) \cite{houska2016augmented}, which decomposes the problem by independently optimizing local JT-UE variables across distributed cores and then coordinating them via a central aggregation step. However, such consensus-based algorithms are known to often require excessive iterations to converge, incurring substantial communication overhead and latency that render them impractical for scheduling \cite{xu2025distributed}. 
	
	Considering these limitations, we instead devise a new distributed scheduling framework that not only fully utilizes the computation resources at the cores but also takes only one round of information exchange between $\text{PU}_0$ and other PUs. 
	As shown in Fig. \ref{fig:framework}, our proposed framework consists of three stages as follows, while the information exchange in the framework is illustrated in Fig. \ref{fig:Info ex}. 

	\begin{figure}[t] 
		\centering	{\includegraphics[width=0.46\textwidth]{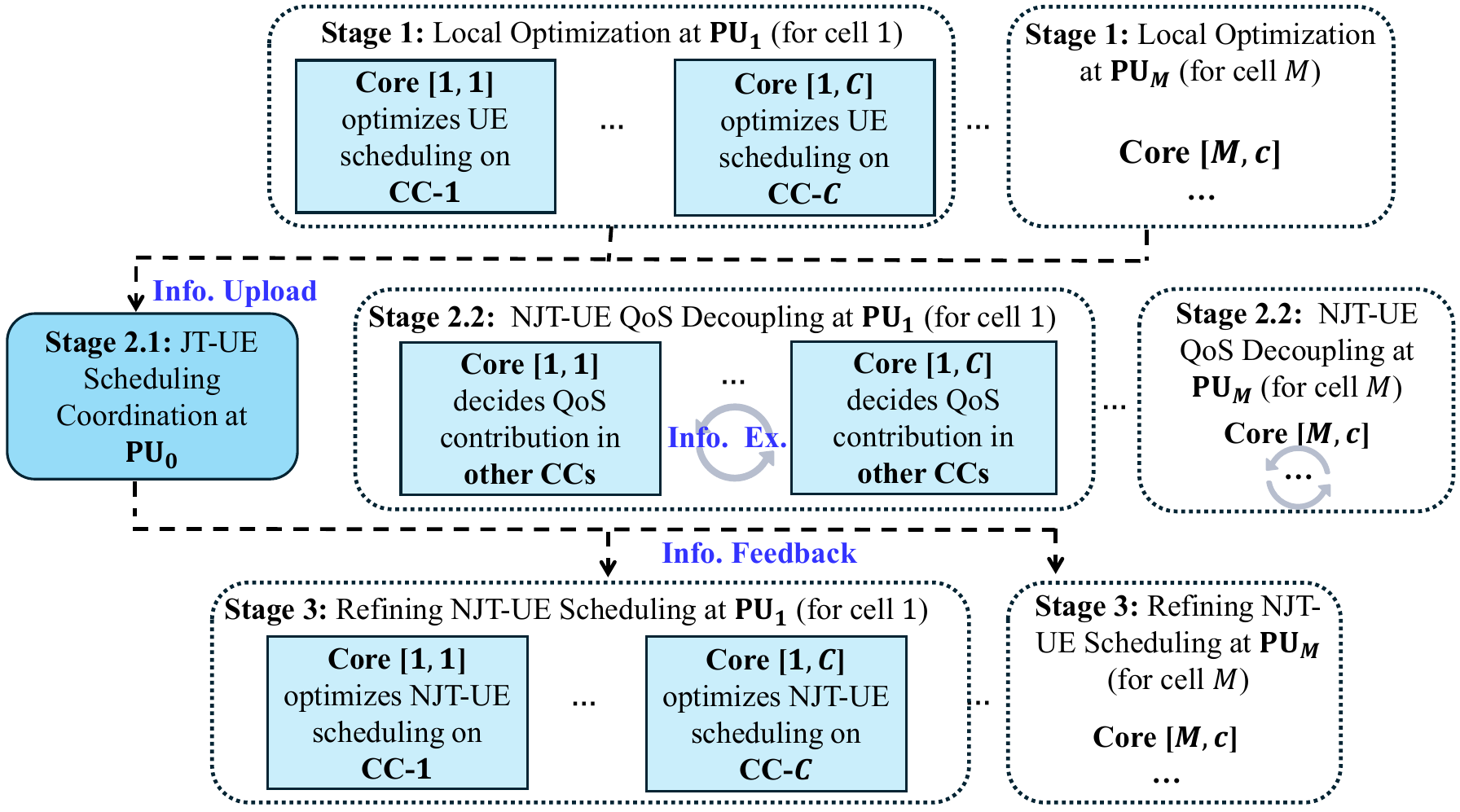}}	
		\caption{Proposed distributed scheduling framework. } \vspace{-0.5cm}
		\label{fig:framework} 
	\end{figure}  
	
	\begin{figure}[t] 
		\centering	        {\includegraphics[width=0.4\textwidth]{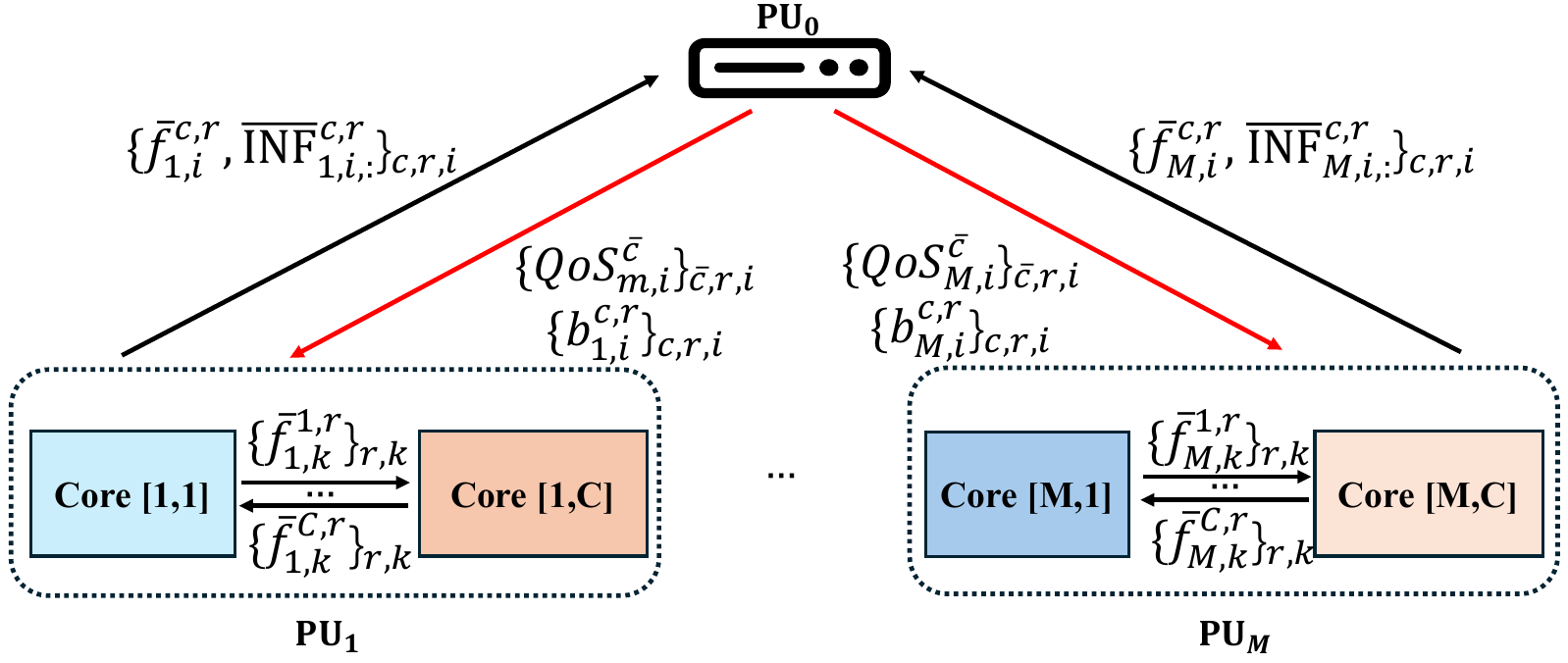}}	
		\caption{Information exchange between nodes. } \vspace{-0.3cm}
		\label{fig:Info ex} 
	\end{figure}
	
	\textbf{Stage 1: Local Optimization at $\text{PU}_1$ to $\text{PU}_M$}.  In this stage, each core within the PUs independently optimizes the scheduling variables on its corresponding CC using only local information. This serves as the initial step of distributed decision-making. The local optimization across PUs is executed in parallel. Upon completion, each PU uploads its optimized results to $\text{PU}_0$, and simultaneously, its internal cores share optimized results with each other.
	
	\textbf{Stage 2: JT-UE Scheduling Coordination at $\text{PU}_0$ and concurrent NJT-UE QoS Decoupling at $\text{PU}_1$ to $\text{PU}_M$}.
	\begin{itemize}
		\item \textit{Stage 2.1}: $\text{PU}_0$ coordinates the scheduling of JT-UEs based on the information received from  $\text{PU}_1$ to $\text{PU}_M$. By fusing the uploaded local results, $\text{PU}_0$ resolves potential inconsistencies and ensures coherent joint scheduling decisions across cells.
		\item \textit{Stage 2.2}: Meanwhile, each core within $\text{PU}_1$ to $\text{PU}_M$ uses the scheduling information exchanged from other cores within the same PU to evaluate the rate contributions of RBGs allocated on different CCs. This evaluation enables decoupling the rate constraints of NJT-UEs with QoS requirements and facilitates coordinated scheduling across multiple CCs.
	\end{itemize} 
	After the Stage 2.1 optimization, $\text{PU}_0$ distributes the updated JT-UE scheduling results to the corresponding PUs, while the refined scheduling results for NJT-UEs in Stage 2.2 are shared among the cores within each PU.
	
	\textbf{Stage 3:  Refining NJT-UE Scheduling at $\text{PU}_1$ to $\text{PU}_M$}: Leveraging the updated information from Stage 2.1 (via $\text{PU}_0$ feedback) and Stage 2.2 (QoS decoupling to each core), each core refines the scheduling decisions for its associated NJT-UEs. This stage aims to further enhance the overall network performance through localized adjustments based on the most recent inter-PU and intra-PU information.
	
	The detailed design of each stage is elaborated as follows.
	\vspace{-0.3cm}
	\subsection{Stage 1: Local Optimization at $\text{PU}_1$ to $\text{PU}_M$} 
	
	This stage is conducted at $\text{PU}_1$ to $\text{PU}_M$ in parallel. For illustration, the index $[m,c]$ is used to denote the $c$-th core at PU $m$. To optimize the scheduling over the RBGs in the $c$-th CC for the JT-UEs and NJT-UE associated with cell $m$, core $[m,c]$ needs to solve  
	
	\vspace{-0.2cm}
	
		\begin{align} \label{minpenalty_core} 
			\max\nolimits_{\{b_{m, k}^{c,r}\}_{r, k \in \mathcal{A}_m}}\ G_{m}^c \ \ \  \text { s.t. } \eqref{Binary},
		\end{align}

	\noindent where
	
	\vspace{-0.4cm}
	
		\begin{align}\label{Gms} 
			G_{m}^c \! =\! & \sum_{r}\sum_{k \in (\mathcal{\hat{K}}_m^{Q}\cup\mathcal{\hat{U}}_m^Q)}  \!\!\tilde{f}_{m, k}^{c,r}  +\rho \!\!\!\!\!\! \sum_{z \in (\mathcal{{K}}_m^{Q}\cup\mathcal{{U}}_m^Q)} \!\!\!\!\min \bigg\{\sum_{r} \tilde{f}_{m, z}^{c,r}, Q_z\bigg\}.
		\end{align}

	\noindent Problem \eqref{minpenalty_core} is constructed as a subproblem of the original formulation \eqref{minpenalty} from the local viewpoint of core $[m,c]$ based on its available information. As such, it addresses QoS demands using only local resources without accounting for JT coordination or inter-CC QoS collaboration. Therefore, directly optimizing \eqref{minpenalty_core} could result in system-level allocations that exceed the actual QoS requirements in \eqref{KNF} and \eqref{INF}. To avoid it, we introduce a heuristic selection mechanism with two evaluation metrics. 
    
	First,  like \eqref{eq:gain1}, the gain of scheduling UE $z \in \Ac_m$ on RBG $(c,r)$  $\text{gain}[c,r,m,z]$ is defined as
	
	\vspace{-0.3cm}
	
		\begin{align}
			\!\!\!\!\text{gain}[c,r,m,z] =     G_m^c(b^{c,r}_{m,z}=1) - G_m^c(b^{c,r}_{m,z}=0).
		\end{align} 
	\vspace{-0.4cm}
	
	\noindent Second, the rate increment of assigning RBG $(c,r)$ to UE $z\in\Ac_m$, i.e., $\tilde{f}^{c,r}_{m,z}$, is evaluated. Based on these two metrics $\{\text{gain}[c,r,m,z], \tilde{f}^{c,r}_{m,z}\}$, PU $m$ filters out RBGs with marginal rate contributions, retaining only those that provide significant gains for scheduling. That is, the values of $\{b_{m,z}^{c,r}\}$ are determined by the following criterion
	
	\vspace{-0.35cm}
	
		\begin{align} \label{eq:distr_s1}
			\!\!\! b_{m,z}^{c,r} =
			\begin{cases}
				1 & \text{if } \text{gain}[c,r,m,z] > 0  \ \& \ {\tilde{f}^{c,r}_{m,z}}/{\tilde{f}^{c,\max}_{m,z}} > \alpha \\
				0 & \text{otherwise},
			\end{cases}
		\end{align}

	\noindent where $\tilde{f}^{c,\max}_{m,z} \triangleq \max_r \{\tilde{f}^{c,r}_{m,z}\}$ is the maximum achievable rate of user $z$ among RBGs of CC-$c$, and $\alpha \in [0,1)$ is a threshold. Here, we choose a relative threshold $\alpha$ rather than an absolute one. This is because achievable rates may differ significantly across CCs, and an absolute threshold easily leads to either overly conservative or overly aggressive RBG allocations on certain CCs.  
	By iteratively applying \eqref{eq:distr_s1} to all UEs, the procedure effectively eliminates the scheduling of the RBGs that only bring marginal contribution to UE's rate improvement and avoid excessive RBG allocation. Accordingly, the scheduling optimization at core $[m,c]$ is detailed in Algorithm \ref{algor2}.
	
	With the outputs of Algorithm \ref{algor2}, each core at $\text{PU}_1$ to $\text{PU}_M$ computes the necessary information for coordination in subsequent stages. Specifically, each core $[m,c]$ first evaluates the achievable rates of its associated JT-UEs $\{\bar f_{m,i}^{c,r}\}_{\,i\in\mathcal{U}_m}$ by temporarily setting \(b_{m,i}^{c,r}=1\), irrespective of the scheduling results from Algorithm \ref{algor2}. Moreover, each core $[m,c]$ computes the interference variations caused by scheduling-state changes of its JT-UEs. When $b_{m,i}^{c,r}$ switches from 1 to 0, the values of interference variation are computed based on \eqref{JTSINR1} by 
	
	\vspace{-0.2cm}
	
		\begin{align}
			\overline{\text{INF}}_{m, i,1\rightarrow 0}^{c,r} = & -\sum\nolimits_{k \in \mathcal{A}_m} b_{m, k}^{c,r} d_{m, i, k}^{c,r} \notag \\
			&\!\!\!\!\!\!\!\!\!\!\!\!\!\!\!\!\!\!\!\!+  \sum\nolimits_{z \in \mathcal{A}_m} b_{m,z}^{c,r} \left( \log(  {\phi}_{m}^{c,r}) - \log ( {\phi}_{m}^{c,r} - 1) \right),
			\label{eq:interference}
		\end{align}

	\noindent where the first term accounts for the reduced interference to other users, and the second term reflects the gain from power redistribution. Similarly, for the switch $b_{m,i}^{c,r}: 0\rightarrow1$, the interference variation is computed as 
	
	\vspace{-0.2cm}
	
		\begin{align}
			\overline{\text{INF}}_{m, i,0\rightarrow 1}^{c,r} = & \sum\nolimits_{k \in \mathcal{A}_m} b_{m, k}^{c,r} d_{m, i, k}^{c,r} \notag \\
			&\!\!\!\!\!\!\!\!\!\!\!\!\!\!\!\!\!\!\!\!+  \sum\nolimits_{z \in \mathcal{A}_m} b_{m,z}^{c,r} \left( \log(  {\phi}_{m}^{c,r}) - \log ( {\phi}_{m}^{c,r} + 1) \right).
			\label{eq:interference1}
		\end{align}

	\noindent The values $\{\bar f_{m,i}^{c,r}\}_{\,i\in\mathcal{U}_m}$ and $\{\overline{\text{INF}}_{m, i,:}^{c,r}\}$ will be uploaded to $\text{PU}_0$ for the JT-UE scheduling coordination in Stage 2.1. Meanwhile, each core $[m,c]$ also computes the rates of its active NJT-UEs $\{\bar f_{m,k}^{c,r}| b_{m,k}^{c,r}=1 \}_{\,k\in\mathcal{K}_m}$ determined by Algorithm \ref{algor2}, and exchanges this information locally with other cores in the same PU for the QoS coordination in Stage 2.2.
	
	\begin{algorithm}[t]
		\caption{ Local Scheduling Optimization at Core $[m,c]$}
		\begin{algorithmic}[1] \label{algor2}
			\STATE Initialize $\{b_{m, k}^{c,r} \}$
			\FOR{each UE $z \in \mathcal{A}_m$}
			\FOR{each RBG $(c,r)$}
			\STATE Calculate $\text{gain}[c,r,m,z]$ and rate $\tilde{f}^{c,r}_{m,z}$;    
			\ENDFOR 
			\STATE Find $\tilde{f}^{c,max}_{m,z} = \argmax_r \{\tilde{f}^{c,r}_{m,z}\}$; 
			\FOR{each RBG $r$}
			\IF{$\text{gain}[c,r,m,z] > 0$ and $\tilde{f}^{c,r}_{m,z}/\tilde{f}^{c,max}_{m,z} > \alpha$ }
			\STATE $b_{m, z}^{c,r}=1$, 
			\ELSE
			\STATE $b_{m, z}^{c,r}=0$, 
			\ENDIF   
			\ENDFOR     
			\ENDFOR
			\STATE repeat steps 2-14 until certain conditions are met.
		\end{algorithmic}
	\end{algorithm}	
	\vspace{-0.35cm}
	\subsection{Stage 2.1: JT-UE Scheduling Coordination at $\text{PU}_0$}
	
	Since the scheduling variables of JT-UEs are optimized independently at $\text{PU}_1$ to $\text{PU}_M$ in Stage 1, the scheduling consistency is not guaranteed. To resolve the issue, it is necessary to design a fusion strategy at $\text{PU}_0$ based on the uploaded information $\{\bar{f}^{c,r}_{m,i}\}$  and $\{\overline{\text{INF}}_{m, i, :}^{c,r}\}$. The main principle is to determine a unified scheduling decision for each inconsistent JT-UE across cells based on which option delivers better overall system performance.

	Consider a representative conflict scenario: suppose JT-UE $i$ requires the joint transmission from both cell $m$ and cell $n$. If JT-UE $i$ is scheduled by cell $m$ on RBG $(c,r)$ but is \textit{not} scheduled by cell $n$ on the same RBG, it will violate the joint scheduling consistency. To resolve the conflict, our proposed strategy first quantifies the performance impact of the following alternative scheduling choices:
	
	\begin{enumerate}
		\item \textbf{Scheduling:} If JT-UE $i$ is ultimately scheduled on RBG $(c,r)$, the  gain to the overall network is: $
		\mathrm{F}_{i, 0 \rightarrow 1}^{c,r} = \bar{f}_{n, i}^{c,r} + \overline{\text{INF}}_{n, i,0\rightarrow 1}^{c,r}$, 
		where   $\overline{\text{INF}}_{n, i, 0\rightarrow 1}^{c,r}$ captures the negative impact caused by the interference that would decrease the rate of other UEs scheduled on RBG $(c,r)$.
		\item \textbf{Non-scheduling:} If UE $i$ is \textit{not} scheduled on RBG $(c,r)$, the gain is: $
		\mathrm{F}_{i, 1 \rightarrow 0}^{c,r} = -\bar{f}_{m, i}^{c,r} + \overline{\text{INF}}_{m, i,1\rightarrow 0}^{c,r}$,
		where the change of JT-UE $i$'s state on RBG $(c,r)$ from `scheduled' to `unscheduled' would decrease its transmission rate and alleviate the interference to other UEs.
	\end{enumerate}
	Based on the values $\{\mathrm{F}_{i, :}^{c,r}\}$, for JT-UEs without QoS requirements, $\text{PU}_0$ generates the scheduling decisions based on the following criterion
	\vspace{-0.1cm}
	
		\begin{align}\label{pick1}
			b_{i}^{c,r}= \begin{cases}1, &  \text{if} \  \mathrm{F}_{i, 0 \rightarrow 1}^{c,r}>\mathrm{F}_{i, 1 \rightarrow 0}^{c,r}\\0, &   \text { otherwise.}\end{cases}
		\end{align} \vspace{-0.1cm}

	For JT-UEs with QoS requirements, the above criterion cannot be adopted directly. Instead, for a JT-UE $i \in \Ic^Q$, $\text{PU}_0$ determines its scheduling state across RBGs by solving the following optimization problem to avoid the resource wastage
	
	\vspace{-0.2cm}
	
		\begin{subequations}\label{pick}
			\begin{align}
				\min\nolimits_{\left\{ {b}_i^{c,r}\right\}_{c, r}}&  \sum\nolimits_{(c,r)}  {b}_i^{c,r} \\
				\text { s.t. } & \sum\nolimits_{(c,r)}  {b}_i^{c,r}\left(\sum\nolimits_{m\in \mathcal{B}_i}\bar{f}_{m, i}^{c,r} \right) \geq Q_i, \label{pick_rbg}\\
				& \mathrm{F}_{i, 0 \rightarrow 1}^{c,r}>\mathrm{F}_{i, 1 \rightarrow 0}^{c,r}.\label{pick_gain}
			\end{align}
		\end{subequations}

	\noindent Although problem \eqref{pick} appears complex, its solution is relatively straightforward: select the RBGs that satisfy constraint \eqref{pick_gain} and sort their corresponding contributions $\sum\nolimits_{m\in \mathcal{B}_i}\bar{f}_{m, i}^{c,r}$ in a decreasing order; then, choose the top-performing RBGs among them to meet constraint \eqref{pick_rbg}. It is worth noting that the optimization for each UE can be executed in parallel, as their scheduling decisions have been fully decoupled at this stage, thereby significantly reducing the processing time. 
	After applying \eqref{pick1} and \eqref{pick} to all JT-UEs, $\text{PU}_0$ updates their achievable rates as
	
	\vspace{-0.2cm}
	
		\begin{align} \label{updated}
			\hat{f}_{m,i}^{c,r}= \begin{cases}\bar{f}_{m,i}^{c,r}, &  \text{if} \ b_{i}^{c,r} = 1, m \in \mathcal{B}_i  \\0, &   \text { otherwise.}\end{cases}
		\end{align}

	With the updated scheduling states and rates of JT-UEs, $\text{PU}_0$ further computes the QoS contribution for JT-UE $i$ from all the other cores except core $[m,c]$ as
	
	\vspace{-0.2cm}
	
		\begin{align}\label{QOS-pC}
			\!\!\!\text{QoS}_{m,i}^{\bar{c}} = \sum\nolimits_{c' \neq c}\sum\nolimits_{r \in \mathcal{R}}\hat{f}_{m, i}^{c',r} + \sum\nolimits_{\ell \in \mathcal{B}_i,\ell\neq m}\sum\nolimits_{(c,r)} \hat{f}_{\ell, i}^{c,r},
		\end{align}

	\noindent where the first term quantifies the QoS contribution from the RBGs scheduled on other CCs by cell $m$, and the second one accounts for the QoS contributed by other cells. The values $\{b_{i}^{c,r}\}$ and $\{\text{QoS}_{m,i}^{\bar{c}} \}$ are then distributed to $\text{PU}_m$ for the final refinement in Stage 3.

	\vspace{-0.3cm}
	\subsection{Stage 2.2: NJT-UE QoS Decoupling at $\text{PU}_1$ to $\text{PU}_M$} 
	When $\text{PU}_0$ conducts the scheduling coordination for JT-UEs, the cores within $\text{PU}_1$ to $\text{PU}_M$ can simultaneously process the results generated from Stage 1 to resolve the scheduling coupling across multi-CCs due to the QoS constraints for NJT-UEs.     The key idea of the decoupling is to decompose each UE’s QoS requirement into per-CC contribution similar to \eqref{QOS-pC}: by quantifying how much rate each CC should provide, the QoS requirement is distributed across CCs. 
	
	Specifically,
	with the exchanged $\{\bar f_{m,k}^{c,r}| b_{m,k}^{c,r}=1 \}_{\,k\in\mathcal{K}_m}$ from other cores, the scheduling states across multi-CCs of the NJT-UEs are optimized by solving the following problem
	
	\vspace{-0.2cm}
	
		\begin{subequations}\label{pick3}
			\begin{align} \label{NFB_UE_QoS}
				\min\nolimits_{\left\{ {b}_{m,k}^{c,r}\right\}}& \sum\nolimits_{(c,r)}  {b}_{m,k}^{c,r} \\
				\text { s.t. } & \sum\nolimits_{(c,r)}  {b}_{m,k}^{c,r}\bar{f}_{m, k}^{c,r} \geq Q_{k}. 
			\end{align}
		\end{subequations}

	\noindent Compared with \eqref{pick}, problem \eqref{pick3} does not impose the condition in \eqref{pick_gain}, since the scheduling of NJT-UEs is only related to their serving cell, without cross-cell inconsistencies. Similar to \eqref{pick}, problem \eqref{pick3} can be solved efficiently by sorting the rate contributions $\{\bar{f}_{m, k}^{c,r}\}$ and only selecting these high-contribution RBGs to meet the QoS requirement. 
	Moreover, optimizing \eqref{pick3} tends to change some RBGs to an “unscheduled” state, alleviating the possible excessive RBG allocations in Stage 1 due to the limitations of local optimization. Meanwhile, releasing the RBGs previously allocated to NJT-UEs benefits the QoS performance of co-scheduled JT-UEs, as the interference they experience is reduced.
	
	After solving \eqref{pick3}, $\text{PU}_m$ updates the achievable rates $\{\hat{f}_{m,k}^{c,r}\}_{k \in \Kc_m^Q}$ similar to \eqref{updated}. Then, the QoS contribution for an NJT-UE $k \in \Kc_m^Q$  from all the other cores except core $[m,c]$ is computed as $\text{QoS}_{m,k}^{\bar{c}} = \sum\nolimits_{c' \neq c } \sum\nolimits_{r \in \mathcal{R}} \, \hat f_{m,k}^{c',r}$, which will be used for the parallel refinement in Stage 3.
	%
	%
	
    
	\vspace{-0.3cm}
	\subsection{Stage 3: Refining NJT-UE Scheduling at $\text{PU}_1$ to $\text{PU}_M$}
	After completing Stage 2.1, $\text{PU}_0$ feeds back the updated scheduling states $\{b^{c,r}_{m,i}\}$ and the associated QoS contribution $\{\text{QoS}^{\bar{c}}_{m,i}\}$ of the JT-UEs to  $\text{PU}_m$. With the feedback, $\text{PU}_m$ further refines the scheduling of NJT-UEs associated with cell $m$. 
	Moreover, with the QoS contributions from other CCs $\{\text{QoS}_{m,k}^{\bar{c}}\}_{k \in \Kc_m^Q}$, \eqref{minpenalty} can be fully decoupled across cores, thereby enabling parallel refinement.
	Specifically, each core $[m, c]$ refines the scheduling of NJT-UEs  $k\in \Kc_m$ on CC-$c$ by solving a sub-problem of \eqref{minpenalty} as
	
    \vspace{-0.3cm}
	 
        \begin{equation} \notag
		\begin{aligned} 
		\underset{\{b_{m, k}^{c,r}\}_{r, k}}{\max}	& \! \tilde{G}_{m}^c  \! \triangleq \! \sum\nolimits_{k \in \mathcal{\hat{K}}_m^Q} \! \sum\nolimits_{r \in \mathcal{R}} \tilde{f}_{m, k}^{c,r}\! + \!\sum\nolimits_{i \in \mathcal{\hat{U}}_m^{Q}} \! \sum\nolimits_{r \in \mathcal{R}} \tilde{f}_{m, i}^{c,r}     \\[-0.5ex]
			& \!\!\!\! \!\!\!\! \!\!\!\! +  \rho \sum\nolimits_{z \in (\mathcal{K}_m^{Q}\cup\mathcal{U}_m^{Q})} \! \min \left\{ \sum\nolimits_{r \in \mathcal{R}}\tilde{f}_{m, z}^{c,r} + \text{QoS}_{m, z}^{\bar{c}}, Q_{z}\right\}. 
		\end{aligned}
        \end{equation}

	\noindent Here, we refrain from re-optimizing the JT-UE scheduling to avoid introducing inconsistencies across different cells and additional rounds of coordination to resolve them. While the adjustment of NJT-UE scheduling may affect the rates of JT-UEs, their QoS requirements can be actively reinforced with a proper penalty parameter $\rho$. Since the above problem shares a similar structure with \eqref{minpenalty}, it can be efficiently solved in a BCD manner. With this refinement, this stage further enhances the overall network performance. Eventually, a distributed scheduling scheme aligned with the O-RAN structure is established by integrating the above three stages.
	
	\vspace{-0.1cm}
	\section{Simulation Results} \label{sec:simu}
	In the simulation, we consider an O-RAN with 3 cells. The O-RUs, each equipped with $N_t = 64$ antennas unless otherwise specified, are located at \([0, -300, 25]\), \([-1000, -300, 25]\), and \([-500, -1200, 25]\) (unit: m). A total of 
	$K$ UEs (height = $1.5~\text{m}$) equipped with 
	$N_r = 4$ antennas are randomly distributed in the rectangular region 
	$[x_{min}, x_{max}, y_{min}, y_{max}] = [-1400, 400, -1400, -100]$ (unit: m). The system uses three CCs centered at 3.2~GHz, 3.5~GHz, and 3.8~GHz. Each CC comprises 624 subcarriers grouped into 13 RBGs of 48 subcarriers each. $P = 10$~dBm and the power spectral density of noise is  $-174$~dBm/Hz. Channels are generated using Quadriga  \cite{jaeckel2014quadriga}. The O-RU–UE association follows a channel-strength-based strategy \cite{tanbourgi2014analysis}, where each UE is primarily connected to the O-RU with the strongest large-scale channel gain. To enable JT, a UE is also associated with other O-RUs if the difference between their channel gains and the strongest one is within $10$ dB.
	Among the 
	$K$ UEs, $K_q \le K$ are randomly selected to have heterogeneous QoS demands, with thresholds $\{Q_k, Q_i\}$ randomly drawn from $[0, 60]~ \text{bps/Hz}$.
	We evaluate the proposed designs using two performance metrics: the effective sum rate (ESR) $G_{esr}$ defined as
	
	\vspace{-0.3cm}
	 
        \begin{equation} \notag
		\begin{aligned} 
			&G_{esr} \triangleq \sum\limits_{i \in\mathcal{\hat{I}}^Q}  \sum\limits_{(c,r)}  {f}_{ i}^{c,r}    + \sum\nolimits_{i \in \mathcal{I}^{Q}} \min \big\{ \sum\limits_{(c,r)}  {f}_{i}^{c,r}, Q_i\big\} \\
			&+\sum\limits_{m \in \mathcal{M}}\bigg(\sum\limits_{k \in \mathcal{\hat{K}}_m^Q}\sum\limits_{(c,r)}  {f}_{m, k}^{c,r} \notag +  \sum\limits_{k \in \mathcal{K}_m^{Q}} \min \big\{ \sum\limits_{(c,r)}  {f}_{m, k}^{c,r}, Q_{ k} \big\}\bigg) 
		\end{aligned}
        \end{equation}

	\noindent where $\{f_{m, k}^{c,r}, f_{ i}^{c,r}\}$ from \eqref{realrate} and  \eqref{realrate11} take the original achievable rates without approximation, and $\mathrm{Sat}$ represents the ratio of UEs whose QoS requirements are met. 
	\vspace{-0.2cm}
	\subsection{Convergence of Proposed Centralized Algorithm}
	\begin{figure}[t] 
		\centering	               {\includegraphics[width=0.35\textwidth]{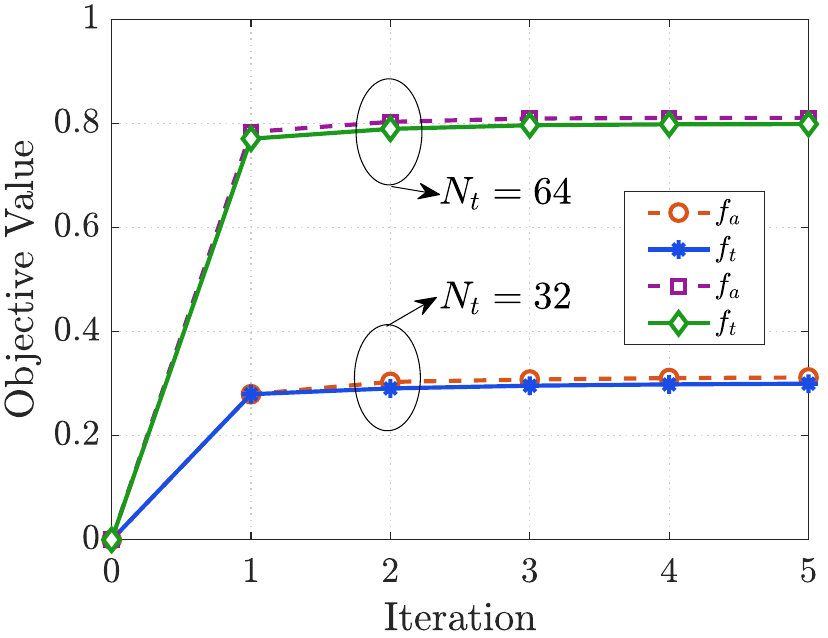}}	
		\vspace{-0.2cm}
		\caption{Convergence of the proposed centralized scheduling. } \vspace{-0.2cm}
		\label{Centralized Algorithm}
	\end{figure}
	
    Fig. \ref{Centralized Algorithm} shows the objective value of problem \eqref{minpenalty} achieved by the proposed centralized scheduling (PCS) Algorithm \ref{algor1} versus iterations, with $K=45$ and $K_q = 25$. Specifically, \( f_{{a}} \triangleq G/ (CRK) \), with $G$ in \eqref{G} with $\rho =1$, denotes the averaged approximate ESR per RBG per UE, while \( f_{{t}} \triangleq G_{{esr}} / (CRK) \) is the attained averaged ESR without approximation. Fig. \ref{Centralized Algorithm} presents that the value of  \( f_{{a}} \) increases rapidly and gets stable within about 5 iterations, indicating the fast convergence of our proposed design. Moreover, \( f_t\) closely follows that of $f_a$, with a relative error consistently below 3\% for $N_t = 32$ and 64, validating the effectiveness of our proposed approximation in Sec. \ref{sec:reformu}.

    \begin{figure*}[t]
    	\centering
    	\begin{minipage}{0.32\textwidth} 
    		\includegraphics[width=\linewidth]{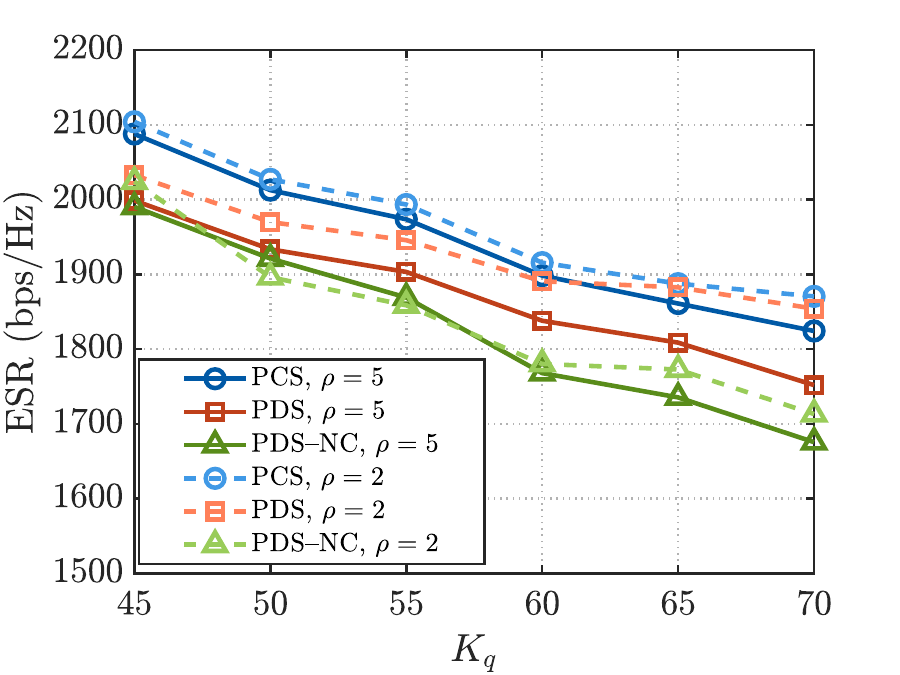}
    		\caption{ESR versus the number of QoS-constrained UEs $K_q$, with $K = 75$.} 
    		\label{fig:rho}
    	\end{minipage}
    	\hfill
    	\begin{minipage}{0.32\textwidth} 
    		\includegraphics[width=\linewidth]{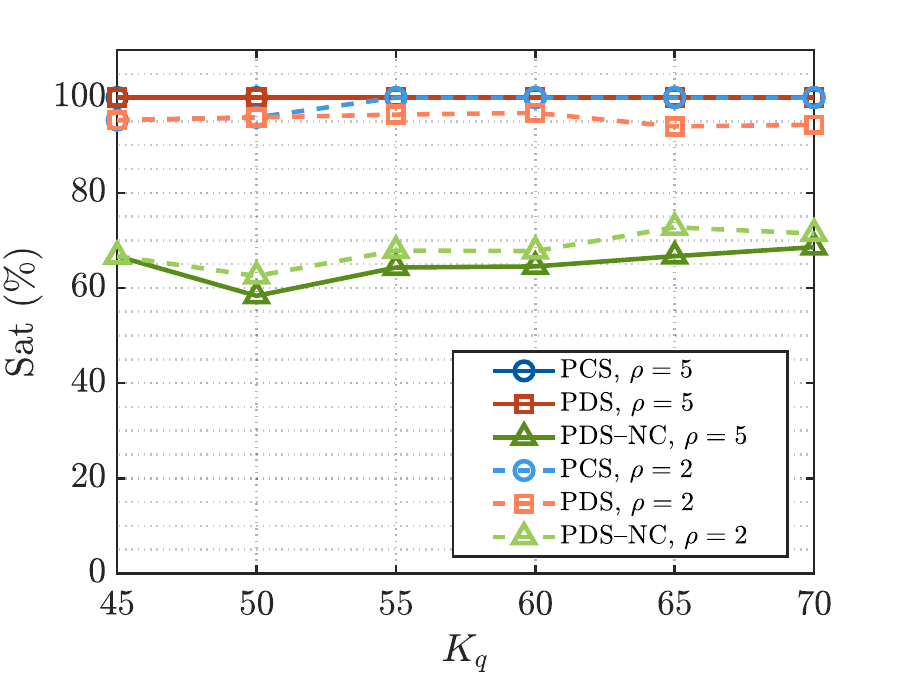} 
    		\caption{$\mathrm{Sat}$ versus the number of QoS-constrained UEs $K_q$, with $K = 75$.}
    		\label{fig:rho_1}
    	\end{minipage}
    	\hfill
    	\begin{minipage}{0.32\textwidth} 
    		\includegraphics[width=\linewidth]{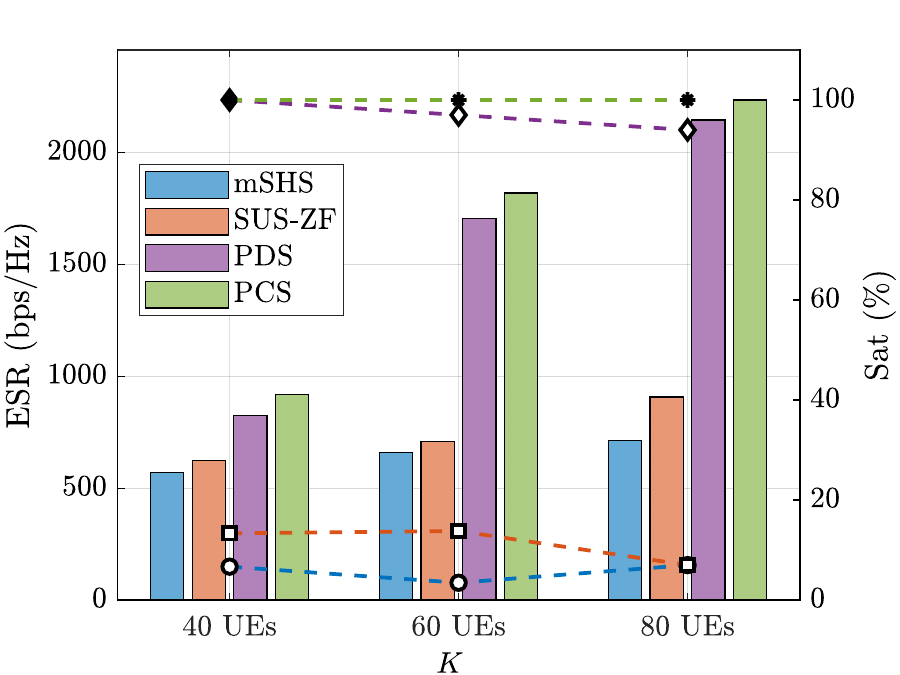} 
    		\caption{ESR (bars) and $\mathrm{Sat}$ (dashed lines) achieved by different schemes versus $K$.}
    		\label{fig:comp}
    	\end{minipage}
    	\vspace{-0.3cm}
    	\label{fig:3wide_independent}
    \end{figure*}

	\vspace{-0.2cm}
	\subsection{Performance of Proposed Scheduling Schemes}
	
	Fig.~\ref{fig:rho} and Fig.~\ref{fig:rho_1} compare the performance of our proposed distributed scheduling (PDS) with that of PCS under different penalty parameters ($\rho$) values.
    The performance of PDS without PU coordination (Stage 2.2 in Sec. \ref{sec:decentralized}-C, named as ``PDS-NC'') is also presented to evaluate the improvement enabled by this step.

	As shown in Fig.~\ref{fig:rho}, the ESR achieved by all schemes decreases as $K_q$ increases. This is expected since more spectrum resources will be allocated to meet the stringent QoS requirements. From an optimization perspective, more QoS-constrained UEs further restrict the feasible set of problem \eqref{SRmax-original}, yielding a smaller ESR. Among the three schemes, PCS performs best due to its global coordination. However, this comes at the cost of intensive computation at a single node, limiting its applicability in large-scale networks. While PDS achieves a lower ESR than PCS,  it distributes the computational load across multiple nodes with only one round of coordination, greatly improving scalability. 
    Compared with PDS, the ESR of PDS-NC gets worse. This stems from its uncoordinated scheduling: each core makes isolated RBG decisions with intra-PU coordination, often resulting in excessive resources for some UEs and insufficiency for others.

    Additionally, Fig.~\ref{fig:rho_1} shows that both PCS and PDS maintain high $\mathrm{Sat}$ values across all $\rho$ values, while PDS-NC achieves significantly lower $\mathrm{Sat}$ (below 78\% in all cases) due to its lack of PU coordination. Combined with Fig.~\ref{fig:rho}, these results demonstrate the critical role of coordination in distributed scheduling. Further, comparing $\rho = 2$ and $5$ reveals the penalty parameter's trade-off effect: smaller $\rho$ improves ESR while larger $\rho$ enhances $\mathrm{Sat}$, as a stronger penalty on QoS violations (see \eqref{G}) shifts optimization priority from throughput to constraint satisfaction. Among the three schemes, PCS shows the least sensitivity to $\rho$ variation, while PDS maintains $\mathrm{Sat}$ above 93\% even when $\rho$ decreases from 5 to 2.
    
	\begin{table}[t]
		\centering
		\caption{Computational Time (s)}
		\begin{tabular}{l|llll}\hline
			Scheme & $K= 40$ & $K=60$ & $K = 80$  \\ \hline
			PCS    & 559 & 1042 & 1450         \\
			PDS    & 18  & 40   & 64          \\ \hline
		\end{tabular}
		\label{computation1} 
	\end{table}
    
	In Table \ref{computation1}, the computational time of our proposed PCS and PDS is compared. One can observe that the computation of PDS is much faster than that of PCS. In all scenarios, the computation time of PDS is under $1/20$ of PCS's. This owes to the parallel processing enabled by our distributed design, making PDS a scalable solution for large networks.

	\vspace{-0.4cm}
	\subsection{Comparison with Other Schemes}


    In this subsection, we compare our proposed designs with two existing schemes under varying numbers of UEs $K$: SUS-ZF \cite{li2015multicell} with ZF-TBF, and a modified SINR-based heuristic scheme (\textbf{mSHS}), which modifies SHS \cite{idrees2022throughput} by weighting the user selection according to individual QoS demands. Fig. \ref{fig:comp} shows that PCS consistently achieves the highest ESR with 100\% $\mathrm{Sat}$ across all scenarios. In contrast, mSHS fails to satisfy QoS demands due to insufficient RBGs per UE; its aggressive QoS prioritization also hinders the scheduling of UEs with good channels, leading to the lowest ESR and satisfaction. SUS-ZF outperforms mSHS but falls short of our methods, especially as $K$ grows. Meanwhile, PDS closely follows the trend of PCS, achieving over 95\% of its ESR and maintaining $\mathrm{Sat}$ above 94\%, confirming the effectiveness.
	
	\vspace{-0.1cm}
	\section{Conclusion}
	\vspace{-0.1cm}
	In this work, we have investigated the distributed user scheduling problem in MU-MIMO O-RAN  serving both JT and NJT UEs with heterogeneous QoS requirements. To handle the complex rate expression involving strong variable coupling across carriers, users, and cells, we have extended the EZF-TBF to JT-UEs, and proposed a novel reformulation scheme that achieves a tractable approximate rate suitable for distributed optimization. Building upon this reformulation, in addition to a BCD-based centralized algorithm, we have proposed a three-stage `decomposition -- coordination -- refinement' distributed scheduling scheme fully tailored to the O-RAN architecture. This distributed scheme only takes a single round of information exchange between nodes, making it suitable for practical latency- and fronthaul-constrained deployments. Simulation results demonstrate that the proposed distributed scheme achieves performance comparable to the centralized benchmark while significantly outperforming existing designs, providing a scalable and efficient solution for large-scale O-RAN deployments with diverse QoS demands.

	\footnotesize 
	\bibliographystyle{IEEEtran}
	\bibliography{refs_journal_MCRA}

@ARTICLE{10969847,
  author={Agarwal, Bharat and Irmer, Ralf and Lister, David and Muntean, Gabriel-Miro},
  journal={IEEE Commun. Surv. Tutor.}, 
  title={{Open RAN for 6G Networks: Architecture, Use Cases and Open Issues}}, 
  year={2026},
  volume={28},
  number={},
  pages={2881-2924}
  }

@INPROCEEDINGS{9488684,
  author={Chen, Yongce and Wu, Yubo and Hou, Y. Thomas and Lou, Wenjing},
  booktitle={Proc. IEEE INFOCOM}, 
  title={{mCore: Achieving Sub-millisecond Scheduling for 5G MU-MIMO Systems}}, 
  year={2021},
  month     = {Jul.},
  address   = {Vancouver, BC, Canada},
  pages={1-10}
}

@article{jager2020blockwise,
  title={The blockwise coordinate descent method for integer programs},
  author={J{\"a}ger, Sven and Sch{\"o}bel, Anita},
  journal={Mathematical Methods of Operations Research},
  volume={91},
  number={2},
  pages={357--381},
  year={2020},
  publisher={Springer}
}

@article{houska2016augmented,
  title={{An augmented Lagrangian based algorithm for distributed nonconvex optimization}},
  author={Houska, Boris and Frasch, Janick and Diehl, Moritz},
  journal={SIAM J OPTIMIZ},
  volume={26},
  number={2},
  pages={1101--1127},
  year={2016},
  publisher={SIAM}
}

@inproceedings{kaziu2024approximate,
  author    = {Brikena Kaziu and Nikita Shanin and Danilo Spano and Li Wang and Wolfgang Gerstacker and Robert Schober},
  title     = {{Approximate Partially Decentralized Linear EZF Precoding for Massive MU-MIMO Systems}},
  booktitle = {Proc. {IEEE} {VTC-Fall}},
  pages     = {1--6},
  month     = {Oct.},
  year      = {2024},
  address   = {Los Angeles, CA, USA}
}

@book{marsch2011coordinated,
  title={Coordinated Multi-Point in Mobile Communications: From Theory to Practice},
  author={Marsch, Patrick and Fettweis, Gerhard P},
  year={2011},
  publisher={Cambridge University Press}
}

@article{xu2025distributed,
  title={Distributed signal processing for extremely large-scale antenna array systems: State-of-the-art and future directions},
  author={Xu, Yanqing and Larsson, Erik G and Jorswieck, Eduard A and Li, Xiao and Jin, Shi and Chang, Tsung-Hui},
  journal={IEEE J. Sel. Top. Signal Process},
  year={2025},
  publisher={IEEE}
}

@article{jaeckel2014quadriga,
  title={{QuaDRiGa: A 3-D Multi-cell Channel Model with Time Evolution for Enabling Virtual Field Trials}},
  author={Jaeckel, Stephan and Raschkowski, Leszek and B{\"o}rner, Kai and Thiele, Lars},
  journal={IEEE Trans. Antennas Propag.},
  volume={62},
  number={6},
  pages={3242--3256},
  year={2014},
  publisher={IEEE}
}

@article{he2022cross,
  author  = {Shiwen He and Zhenyu An and Jianyue Zhu and Min Zhang and Yongming Huang and Yaoxue Zhang},
  title   = {{Cross-Layer Optimization: Joint User Scheduling and Beamforming Design With QoS Support in Joint Transmission Networks}},
  journal = {IEEE Trans. Commun.},
  volume  = {71},
  number  = {2},
  pages   = {792--807},
  year    = {2022}
}

@article{jorswieck2009throughput,
  author  = {Eduard A. Jorswieck and Aydin Sezgin and Xi Zhang},
  title   = {{Throughput Versus Fairness: Channel-Aware Scheduling in Multiple Antenna Downlink}},
  journal = {EURASIP J. Wirel. Commun. Netw.},
  volume  = {2009},
  number  = {1},
  pages   = {271540},
  year    = {2009},
  publisher = {Springer}
}

@inproceedings{bischoff2024real,
  author    = {Tano Bischoff and Martin Kasparick and Ehsan Tohidi and S{\l}awomir Sta{\'n}czak},
  title     = {{Real-time Algorithms for Combined {eMBB} and {URLLC} Scheduling}},
  booktitle = {Proc. {WSA}},
  pages     = {1--5},
  month     = {Mar.},
  year      = {2024},
  address   = {Nuremberg, Germany},
}

@inproceedings{tanbourgi2014analysis,
  author    = {Ralph Tanbourgi and Sarabjot Singh and Jeffrey G. Andrews and Friedrich K. Jondral},
  title     = {{Analysis of {Non-Coherent} {Joint-Transmission} Cooperation in {Heterogeneous Cellular Networks}}},
  booktitle = {Proc. {IEEE} {ICC}},
  pages     = {5160--5165},
  month     = {Jun.},
  year      = {2014},
  address   = {Sydney, Australia},
  doi       = {10.1109/ICC.2014.6884100}
}

@article{chen2023om,
  author  = {Yongce Chen and Y. Thomas Hou and Wenjing Lou and Jeffrey H. Reed and Sastry Kompella},
  title   = {{{OM}$^3$: Real-Time Multi-Cell {MIMO} Scheduling in {5G O-RAN}}},
  journal = {IEEE J. Sel. Areas Commun.},
  volume  = {42},
  number  = {2},
  pages   = {339--355},
  year    = {2023},
  doi     = {10.1109/JSAC.2023.3240140}
}

@article{idrees2022throughput,
  author  = {Muhammad Idrees and Xiaomin Qi and Alam Zaib and Adnan Tariq and Irfan Ullah and Zahid Mahmood and Shahid Khattak},
  title   = {{Throughput Maximization in Clustered Cellular Networks by Using Joint Resource Scheduling and Fractional Frequency Reuse-Aided Coordinated Multipoint}},
  journal = {Arab. J. Sci. Eng.},
  pages   = {1--15},
  year    = {2022},
  doi     = {10.1007/s13369-022-07162-w}
}

@article{gamvrelis2022slinr,
  author  = {Tyler Gamvrelis and Zehua Li and Ahmad Ali Khan and Raviraj S. Adve},
  title   = {{{SLINR}-Based Downlink Optimization in {MU-MIMO} Networks}},
  journal = {IEEE Access},
  volume  = {10},
  pages   = {123956--123970},
  year    = {2022},
  doi     = {10.1109/ACCESS.2022.3224564}
}

@article{li2015multicell,
  author  = {Min Li and Iain B. Collings and Stephen V. Hanly and Chunshan Liu and Philip Whiting},
  title   = {{Multicell Coordinated Scheduling with Multiuser Zero-Forcing Beamforming}},
  journal = {IEEE Trans. Wireless Commun.},
  volume  = {15},
  number  = {2},
  pages   = {827--842},
  year    = {2015},
  doi     = {10.1109/TWC.2015.2487298}
}

@inproceedings{li2014multicell,
  author    = {Min Li and Chunshan Liu and Iain B. Collings and Stephen V. Hanly},
  title     = {{Multicell Coordinated Scheduling with Multiuser {ZF} Beamforming}},
  booktitle = {Proc. {IEEE} {ICC}},
  pages     = {5006--5011},
  month     = {Jun.},
  year      = {2014},
  address   = {Sydney, Australia},
  doi       = {10.1109/ICC.2014.6884083}
}

@inproceedings{chen2022m,
  author    = {Yongce Chen and Y. Thomas Hou and Wenjing Lou and Jeffrey H. Reed and Sastry Kompella},
  title     = {{{M}$^3$: A Sub-Millisecond Scheduler for Multi-Cell {MIMO} Networks under {C-RAN} Architecture}},
  booktitle = {Proc. IEEE INFOCOM},
  pages     = {130--139},
  year      = {2022},
  doi       = {10.1109/INFOCOM48880.2022.9796733}
}

@inproceedings{sun2010genetic,
  author    = {Fanglei Sun and Mingli You and Jin Liu and Zhenning Shi and Pingping Wen and Jianguo Liu},
  title     = {{Genetic Algorithm Based Multiuser Scheduling for Single- and Multi-Cell Systems with Successive Interference Cancellation}},
  booktitle = {Proc. {IEEE} {PIMRC}},
  pages     = {1230--1235},
  month     = {Sep.},
  year      = {2010},
  address   = {Istanbul, Turkey},
  doi       = {10.1109/PIMRC.2010.5671794}
}

@article{denis2021improving,
  author  = {Juwendo Denis and Mohamad Assaad},
  title   = {{Improving Cell-Free Massive {MIMO} Networks Performance: A User Scheduling Approach}},
  journal = {IEEE Trans. Wireless Commun.},
  volume  = {20},
  number  = {11},
  pages   = {7360--7374},
  year    = {2021},
  doi     = {10.1109/TWC.2021.3098507}
}

@inproceedings{yu2012coordinated,
  author    = {Lei Yu and Eleftherios Karipidis and Erik G. Larsson},
  title     = {{Coordinated Scheduling and Beamforming for Multicell Spectrum Sharing Networks Using Branch \& Bound}},
  booktitle = {Proc. {EUSIPCO}},
  pages     = {819--823},
  month     = {Aug.},
  year      = {2012},
  address   = {Bucharest, Romania}
}

@article{khan2020optimizing,
  author  = {Ahmad Ali Khan and Raviraj S. Adve and Wei Yu},
  title   = {{Optimizing Downlink Resource Allocation in Multiuser {MIMO} Networks via Fractional Programming and the Hungarian Algorithm}},
  journal = {IEEE Trans. Wireless Commun.},
  volume  = {19},
  number  = {8},
  pages   = {5162--5175},
  year    = {2020},
  doi     = {10.1109/TWC.2020.2993973}
}

@article{antonioli2020decentralized,
  author  = {Roberto Pinto Antonioli and G\'abor Fodor and Pablo Soldati and Tarcisio Ferreira Maciel},
  title   = {{Decentralized User Scheduling for Rate-Constrained Sum-Utility Maximization in the {MIMO} {IBC}}},
  journal = {IEEE Trans. Commun.},
  volume  = {68},
  number  = {10},
  pages   = {6215--6229},
  year    = {2020},
  doi     = {10.1109/TCOMM.2020.3014042}
}

@article{polese2023understanding,
  author  = {Michele Polese and Leonardo Bonati and Salvatore D'oro and Stefano Basagni and Tommaso Melodia},
  title   = {{Understanding {O-RAN}: Architecture, Interfaces, Algorithms, Security, and Research Challenges}},
  journal = {IEEE Commun. Surv. Tutor.},
  volume  = {25},
  number  = {2},
  pages   = {1376--1411},
  year    = {2023},
  doi     = {10.1109/COMST.2023.3242547}
}

@article{sun2010eigen,
  author  = {Liang Sun and Matthew R. McKay},
  title   = {{Eigen-Based Transceivers for the {MIMO} Broadcast Channel with Semi-Orthogonal User Selection}},
  journal = {IEEE Trans. Signal Process.},
  volume  = {58},
  number  = {10},
  pages   = {5246--5261},
  year    = {2010},
  doi     = {10.1109/TSP.2010.2053709}
}

@article{zhou2025multi,
  author  = {Zihao Zhou and Cheng-Xiang Wang and Xinyue Chen and Li Zhang and Jie Huang and Lijian Xin and Xiping Wu},
  title   = {{Multi-Frequency Wireless Channel Measurements and Modeling in Urban Macro Scenarios}},
  journal = {IEEE Trans. Veh. Technol.},
  year    = {2025},
  doi     = {10.1109/TVT.2025.3569181} 
}

@inproceedings{lin2022carrier,
  author    = {Pingping Lin and Chunlei Hu and Xiao Li and Jinyang Yu and Weiliang Xie},
  title     = {{Research on Carrier Aggregation of 5G NR}},
  booktitle = {Proc. IEEE BMSB},
  address   = {Bilbao, Spain},
  pages     = {1--5},
  month     = {Jun.},
  year      = {2022},
  doi       = {10.1109/BMSB55706.2022.9828744}
}

@inproceedings{seifi2011coordinated,
  author    = {Nima Seifi and Michail Matthaiou and Mats Viberg},
  title     = {{Coordinated User Scheduling in the Multi-Cell {MIMO} Downlink}},
  booktitle = {Proc. {IEEE} {ICASSP}},
  pages     = {2840--2843},
  month     = {May},
  year      = {2011},
  address   = {Prague, Czech Republic},
  doi       = {10.1109/ICASSP.2011.5947043}
}

@book{nocedal1999numerical,
  title={Numerical optimization},
  author={Nocedal, Jorge and Wright, Stephen J},
  year={1999},
  publisher={Springer}
}

@article{bian2024decentralizing,
  title={{Decentralizing Coherent Joint Transmission Precoding via Fast ADMM with Deterministic Equivalents}},
  author={Bian, Xinyu and Liu, Yuhao and Xu, Yizhou and Hou, Tianqi and Wang, Wenjie and Mao, Yuyi and Zhang, Jun},
  journal={arXiv preprint arXiv:2403.19127},
  year={2024}
}

@misc{vershynin2009high,
  title={High-dimensional probability},
  author={Vershynin, Roman},
  year={2009},
  publisher={Cambridge University Press Cambridge, UK}
}

@ARTICLE{Yubo_TWC24,
  author={Wu, Yubo and Shi, Yi and Hou, Y. Thomas and Lou, Wenjing and Reed, Jeffrey H. and DaSilva, Luiz A.},
  journal={IEEE Trans. Wireless Commun.}, 
  title={{R$^3$: A Real-Time Robust {MU-MIMO} Scheduler for {O-RAN}}}, 
  year={2024},
  volume={23},
  number={11},
  pages={17727-17743},
  doi={10.1109/TWC.2024.3456596}}
	\ifCLASSOPTIONcaptionsoff
	\newpage
	\fi
	
\end{document}